\documentclass{article}
\usepackage{gastex}
\usepackage{amsmath}
\usepackage{amssymb}
\usepackage{amsfonts}
\usepackage{theorem}
\usepackage{url}
\usepackage{pst-all} 
 \psset{unit=0.7cm} 
 \newpsobject{showgrid}{psgrid}{%
   subgriddiv=1,griddots=10,gridlabels=6pt}
 \psset{linewidth=0.01}
\newcommand\treesepvalue{0.6}
\newpsstyle{treestyle}{unit=0.4cm,
  linewidth=0.05,
  treesep=\treesepvalue,levelsep=1.5,
  tnsep=0.2, 
  fillstyle=solid,
  linecolor=blue,
  radius=0.25,framesize=0.45,framesep=0}
\newpsstyle{simpletreestyle}{
  treemode=R,
  linewidth=0.02,
  treesep=1,
  levelsep=1.6,
  tnsep=0.0,
  tnpos=a, 
  radius=0.1}
\newpsstyle{fulltree}{unit=0.5cm,
  linewidth=0.02,
  treesep=1,
  levelsep=2,
  tnsep=0.2,
  fillstyle=solid}
\newcommand\MTu[3][\treesepvalue]{\pstree[thistreesep=#1]{\TR{\,$#2$\,}}{#3}}

\newcommand\MTp[2][\treesepvalue]{\pstree[thistreesep=#1]{\Tp}{#2}}

\usepackage{theorem}
\newtheorem{theorem}{Theorem}

\newtheorem{proposition}[theorem]{Proposition}
\newtheorem{corollary}[theorem]{Corollary}
{\theorembodyfont{\rmfamily}%
   \newtheorem{example}{\sc Example}

}
\numberwithin{theorem}{subsection}
\numberwithin{equation}{section}
\numberwithin{figure}{section}
\numberwithin{table}{section}
\newcommand\A{\mathcal{A}}

\newcommand\GD{\mathcal{D}}
\newcommand\GR{\mathcal{R}}
\newcommand\GL{\mathcal{L}}
\newcommand\GH{\mathcal{H}}

\newcommand{\Q}{\mathbb{Q}}
\newcommand{\Z}{\mathbb{Z}}
\newenvironment{proof}{\noindent\textit{Proof}}
                       {\QED\vskip\theorempostskipamount}
\def\petitcarre{\vrule height4pt width 4pt depth0pt}
\def\QED{\relax\ifmmode\eqno{\hbox{\petitcarre}}\else{%
  \unskip\nobreak\hfil\penalty50\hskip2em\hbox{}\nobreak\hfil
  \petitcarre
  \parfillskip=0pt \finalhyphendemerits=0\par\smallskip}
  \fi}
\newenvironment{proofof}[1]{\noindent\textit{Proof
    \protect{#1}.}}
                       {\QED\vskip\theorempostskipamount}
\def\petitcarre{\vrule height4pt width 4pt depth0pt}
\def\QED{\relax\ifmmode\eqno{\hbox{\petitcarre}}\else{%
  \unskip\nobreak\hfil\penalty50\hskip2em\hbox{}\nobreak\hfil
  \petitcarre
  \parfillskip=0pt \finalhyphendemerits=0\par\smallskip}
  \fi}
\newcommand{\edge}[1]{\stackrel{#1}{\rightarrow}}
\def\u(#1){\underline{#1}\:}

\DeclareMathOperator{\Card}{Card}
\DeclareMathOperator{\Dom}{Dom}
\DeclareMathOperator{\End}{End}
\title{Birecurrent sets}
\author{Francesco Dolce$^1$, Dominique Perrin$^2$, \\
Antonio Restivo$^3$, Christophe Reutenauer$^1$,
Giuseppina Rindone$^2$\\
$^1$ Universit\'e du Qu\'ebec \`a Montr\'eal, LaCIM,
$^2$ Universit\'e Paris Est, LIGM,\\ $^3$ Universit\`a di Palermo}
\begin{document}
\maketitle
\begin{abstract}
A  set is called recurrent if its minimal automaton is strongly connected
and birecurrent if it is recurrent as well as
 its reversal. 
We prove a series of results concerning birecurrent sets. 
It is already known that any birecurrent set is completely
reducible (that is, such that
the minimal representation of its characteristic series is completely
reducible). The main result of this paper
 characterizes  completely reducible sets  as linear combinations
of birecurrent sets
\end{abstract}
\tableofcontents
\section{Introduction}
A recurrent set of words is such that its minimal automaton is strongly
connected. A birecurrent set is a recurrent set such that its
reversal is also recurrent.  The recurrent (resp. birecurrent) sets contains the submonoids generated by prefix (resp. bifix) codes.
The automata recognizing birecurrent are a generalization
of well known families of automata such as group automata.
There are more general sets and the general form of  birecurrent sets
does seem easy to describe.

Our interest in birecurrent sets is motivated by their relation
with completely reducible sets put in evidence 
in \cite{Perrin2013}. 
Indeed, it is shown in \cite{Perrin2013} that the syntactic representation
of the characteristic series of
a birecurrent set is completely reducible. This generalizes
the result of Reutenauer \cite{Reutenauer1981} which proves that 
the complete reducibility holds
for the submonoid generated by a bifix code.

Our main result 
 characterizes completely reducible sets as linear
combinations of birecurrent sets (Theorem~\ref{theoremCR3}).
We also prove a number of other results concerning birecurrent sets,
and in particular birecurrent sets of finite type,
which are a generalization of the submonoids generated
by finite bifix codes.

The paper is organized as follows.

We first give in Section~\ref{sectionPreliminaries} a number of
definitions concerning words, automata and formal series.

We introduce recurrent and birecurrent sets 
in Section~\ref{sectionBirecurrent}.
In Section~\ref{sectionRev}
we prove  a result which characterizes 
 the minimal automata of birecurrent sets (Theorem~\ref{theoremRev}). This result extends
a property  proved in~\cite{Perrin2013}
and allows to construct directly birecurrent automata by choosing
an appropriate set of terminal states of a strongly connected deterministic automaton.

In Section~\ref{sectionFiniteType}, we define birecurrent sets of finite type.
A recurrent set $S$
is of the form $S=X^*P$ where $X$ is a prefix code
and $P$ a set of proper prefixes of $X$. Thus a birecurrent
set $S$ has the form $X^*P$, as above, as well as
$S=QY^*$ where $Y$ is a suffix code and $Q$ a set of proper
suffixes of $Y$. We say that $S$ is of finite type if $X$ and $Y$
are finite. Thus a birecurrent set of finite type is a generalization
of a submonoid generated by a finite bifix code
(which corresponds to the case where $P$ and $Q$
are reduced to the empty word). We prove a result allowing one to build birecurrent
sets of finite type (Theorem~\ref{theoremDP}).

In Section~\ref{sectionDegree}, we define the degree and the index of a dense
birecurrent set. We prove that the density of a dense birecurrent set
with respect to a positive Bernoulli distribution is the inverse
of its index (Theorem~\ref{theoremDensity}).

In Section~\ref{sectionIndecomposable}, we prove that if a recognizable
maximal prefix code is indecomposable, then either it is synchronized,
or it is the left root of a dense birecurrent set (Theorem~\ref{theoremDecomposition}). We relate this result with an old conjecture of Sch\"utzenberger.

In Section~\ref{sectionCompleteReducibility}, we come to the connection with complete reducibility.

We start in Section~\ref{sectionFormalSeries} with an introduction to
  the linear
representations of formal series. In Section~\ref{sectionComplelyReducible},
we prove a statement (Theorem~\ref{theoremCharactCompleteRed}) which 
characterizes completely reducible sets 
by a property of their syntactic monoid. 
We derive from this result several corollaries and, in particular, the main result of ~\cite{Perrin2013}
asserting that a birecurrent set is completely reducible.

We then present in Section~\ref{sectionCharacterizationCompletelyReducible}
our main result which characterizes completely reducible sets as linear
combinations of birecurrent sets (Theorem~\ref{theoremCR3}).

In  Section~\ref{sectionUnambiguous}, we consider  unambiguous automata. We characterize the
unambiguous automata recognizing recurrent sets
(Theorem~\ref{theoremCR}). We derive as a corollary a characterization of
unambiguous automata recognizing birecurrent sets (Corollary~\ref{corollaryCR}).
We also give examples of such automata with an iteration of the construction of Section~\ref{sectionConstruction}, using an argument originally developped in~\cite{Vincent1985}. 

\paragraph{Acknowledgements}
The authors wish to thank for their help several colleagues,
including Clelia De Felice, Andrew Ryzhikov and the anonymous referee.

\section{Preliminaries}\label{sectionPreliminaries}
We recall briefly some terminology about words, automata
and, in some more detail, formal series.
We refer to~\cite{BerstelPerrinReutenauer2009} 
or~\cite{RhodesSteinberg2009} for undefined terms.
\subsection{Words}
Let $A$ be a finite alphabet and let $A^*$ be the free
monoid over $A$. The elements of $A^*$ are called \emph{words}
and the subsets of $A^*$ \emph{formal languages}.
We denote by $\varepsilon$ the empty word.

The \emph{reversal} of a word $w=a_1\cdots a_n$ is the word
$\tilde{w}=a_n\cdots a_1$. By extension, the reversal of
$X\subset A^*$ is the set $\tilde{X}=\{\tilde{w}\mid w\in X\}$.

For $u\in A^*$ and $X\subset A^*$, we denote $u^{-1}X=\{v\in A^*\mid uv\in X\}$.
Symmetrically, we denote $Xv^{-1}=\{u\in A^*\mid uv\in X\}$.

A set $S\subset A^*$ is said to be \emph{dense} in $A^*$ if for
any word $v\in A^*$, there are words $u,w\in A^*$ such that $uvw\in S$.
A set which is not dense is said to be \emph{thin}.

A set $S\subset A^*$ is said to be \emph{right dense} in $A^*$ if for
any $v\in A^*$, there is $w\in A^*$ such that $vw\in S$.
\subsection{Automata}
We denote by $\A=(Q,I,T)$ an automaton
with $Q$ as set of states, $I$ as set of initial states and
$T$ as set of terminal states, given by a set of
edges which are triples $(p,a,q)\in Q\times A\times Q$.
The automaton is said to be finite if $Q$ is finite.

The set \emph{recognized} by the automaton $\A=(Q,I,T)$
 is the set of labels of paths from $I$ to $T$.
A set is \emph{recognizable} if it can be recognized by a finite automaton.

The automaton $\A=(Q,I,T)$ is \emph{deterministic}
if $I=\{i\}$ and for each $p\in Q$ and $a\in A$
there is at most one edge $(p,a,q)$. For $p\in Q$ and $a\in A$,
we denote by $p\cdot a$ the unique state $q$ such that
there is an edge from $p$ to $q$ labeled $a$.
Thus, the set recognized by $\A$ is the set of words $w$ such that
$i\cdot w\in T$.
We denote $\A=(Q,i,T)$  and $(Q,i,t)$ if $T=\{t\}$.

The \emph{minimal automaton} of a set $X\subset A^*$ is the deterministic
automaton having for states the nonempty sets $u^{-1}X$ for $u\in A^*$,
with $X$ as initial state and the family $\{u^{-1}X\mid u\in X\}$
as terminal states.

An automaton $\A=(Q,I,T)$ is \emph{trim} if for any $q\in Q$ there
 are words $u,v$ such that there is a path from $I$ to $q$ labeled $u$
and a path from $q$ to $T$ labeled $v$.

A deterministic automaton $\A=(Q,i,T)$ on the alphabet $A$
is \emph{complete} if for any
$q\in Q$ and any $a\in A$ one has $q\cdot a\ne\emptyset$.
A set $S\subset A^*$ is right dense if and only if its minimal
automaton is complete.

For a deterministic automaton $\A=(Q,i,T)$ and a word $w$, 
we denote by $\varphi_{\A}(w)$ the map $q\mapsto q\cdot w$ from $Q$
into $Q$. The monoid $M_\A=\varphi_\A(A^*)$ is the \emph{transition monoid}
of the automaton $\A$. For $m\in M_\A$, denote $pm=q$ if the
image of $p$ by $m$ is $q$, or equivalently if $m=\varphi_\A(w)$
and $p\cdot w=q$. Similarly, for $S\subset Q$,
we denote $mS=\{p\in Q\mid pm\in S\}$.

The \emph{rank} of a word $w$ is the rank of the map $\varphi_\A(w)$,
that is the cardinality of the set $\{p\cdot w\mid p\in Q\}$.
When $\A$ is a strongly connected finite automaton, the image by $\varphi_\A$
of the set of words of minimal nonzero rank is, together with $0$,
 the unique $0$-minimal
ideal $J$ of the monoid $\varphi_\A(A^*)$. It is the union
of all $0$-minimal right (resp. left) ideals. It is formed of $0$
and a regular $\GD$-class
(see~\cite[Corollary 1.12.10]{BerstelPerrinReutenauer2009}). 
Each $\GH$-class of $J$ which is a group is a transitive
permutation group on the common image of its elements
(see~\cite[Theorem 9.3.10]{BerstelPerrinReutenauer2009}).

Let $\A=(Q,i,T)$ be a deterministic automaton.  The \emph{domain} of a word
$w$, denoted $\Dom(w)$, is the set $w\cdot Q$, that is, the set
of $q\in Q$ such that $q\cdot w\ne\emptyset$. The \emph{kernel}
of a word $w$ is the partial equivalence
on $\Dom(w)$ defined by $p\equiv q$ if $p\cdot w=q\cdot w$.

The \emph{degree} of  a strongly connected deterministic automaton $\A$, denoted $d(\A)$, is the
minimal nonzero rank of all words. The automaton is
\emph{synchronized} if its degree is $1$.
\begin{example}\label{exampleAutomaton}
Let $\A$ be the automaton represented in Figure~\ref{figureExampleAutomaton}.
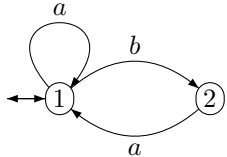
\begin{figure}[hbt]
\centering\gasset{Nadjust=wh}

\begin{picture}(15,20)(0,-5)

\node[Nmarks=if,fangle=180](1)(0,0){$1$}\node(2)(20,0){$2$}

\drawloop(1){$a$}\drawedge[curvedepth=5](1,2){$b$}\drawedge[curvedepth=5](2,1){$a$}
\end{picture}
\caption{A synchronized  automaton.}\label{figureExampleAutomaton}
\end{figure}
The word $b$ is synchronizing since it has rank $1$. The $0$-minimal
ideal of the monoid $\varphi_\A(A^*)$ is represented in Figure~\ref{figureMinimalIdeal0}.
\begin{figure}[hbt]
\begin{displaymath}
\def\rb{\hspace{2pt}\raisebox{0.8ex}{*}}\def\vh{\vphantom{\biggl(}}
    \begin{array}%
    {r|@{}l@{}c|@{}l@{}c|}%
    \multicolumn{1}{r}{}&\multicolumn{2}{c}{1}&\multicolumn{2}{c}{2}\\
    \cline{2-5}
    1,2&\vh\rb &a &\vh\rb  &ab\\
    \cline{2-5}
   1&\vh\rb &ba & & b \\
   \cline{2-5}

    \end{array}
\end{displaymath}
\caption{The $0$-minimal ideal of $\varphi_\A(A^*)$ in Example~\ref{exampleAutomaton}.}\label{figureMinimalIdeal0}
\end{figure}
We represent the $\GD$-class by an array in which 
the rows are the $\GR$-classes and the columns are the $\GL$-classes.
An $\GH$-class containing an idempotent is indicated by a $*$.
\end{example}
\subsection{Prefix and bifix codes}
A set $X\subset A^+$ is a \emph{prefix code} if no word in $X$
is a proper prefix of another word in $X$. The set $X$
is a \emph{bifix code} if $X$ and $\tilde{X}$ are prefix codes.

Let $X\subset A^+$ be a prefix code. The \emph{literal automaton}
of $X^*$ is the deterministic automaton $\A=(Q,\varepsilon,\varepsilon)$ 
where $Q$ is the set
of proper prefixes of $X$  and  for $q\in Q$
and $a\in A$, one has
\begin{displaymath}
q\cdot a=\begin{cases}
qa&\text{if $qa\in Q$}\\
\varepsilon&\text{if $qa\in X$}\\
\emptyset&\text{otherwise.}
\end{cases}
\end{displaymath}
It is a strongly connected deterministic automaton recognizing $X^*$.

The \emph{degree} $d(X)$ of a prefix code $X$ is the degree of the
minimal automaton of $X^*$. The prefix code $X$ is said to
be \emph{synchronized} if $d(X)=1$.
\begin{example}
The automaton of Example~\ref{exampleAutomaton} is both the minimal
and the literal automaton of $X^*$ where $X$ is the synchronized prefix code
$\{a,ba\}$.
\end{example}
\subsection{Formal series}
\def\<<A>>{\langle\langle A\rangle\rangle}
Let $K$ be a field. A \emph{formal series}
on the alphabet $A$
with coefficients in $K$ is a map $\sigma:A^*\rightarrow K$. We denote
by $K\<<A>>$ the algebra of these series.
For
$w\in A^*$, we denote by $(\sigma,w)$ the value of $\sigma$ on $w$.
We denote by $K\langle A\rangle$ the corresponding ring of polynomials,
which are the series with a finite number of nonzero coefficients.

We denote by  $\u(S)$ the characteristic series of a set
$S\subset A^*$.

Let $n\ge 0$ be an integer. Let $\lambda$ be a row $n$-vector, let
$\mu$ be a morphism from $A^*$ into the monoid of $n\times n$-matrices
and let $\gamma$ be a column $n$-vector, all with coefficients in $K$.
The triple $(\lambda,\mu,\gamma)$ is called a \emph{linear representation}.
It is said to \emph{recognize} a series $\sigma$ if for every $w\in A^*$,
one has 
\begin{equation}
(\sigma,w)=\lambda\mu(w)\gamma.\label{eqRepresentation}
\end{equation}
Equivalently, one can define a linear representation as a triple
$(\lambda,\mu,\gamma)$ formed, for some vector space $V$,
 of a vector $\lambda\in V$,
a morphism $\mu:A^*\rightarrow \End(V)$ and a linear form $\gamma$ on $V$
such that Equation~\eqref{eqRepresentation} holds
for every $w\in A^*$. The vector $\lambda$ is called the \emph{initial vector}.
The endomorphism $\mu(w)$ and the linear form
$\gamma$ act on the right of their argument and thus
the expression $\lambda\mu(w)\gamma$ in Equation~\eqref{eqRepresentation}
is read parenthesized from left to right as $(\lambda\mu(w))\gamma$. It can also
be parenthesized from right to left considering $\mu(w)$
as the endomorphism of the dual $V'$ of $V$ defined by the  formula
$\ell(\mu(w)c)=(\ell\mu(w))c$.

\section{Birecurrent sets}\label{sectionBirecurrent}
In this section, we define birecurrent sets and prove some elementary properties.
We prove a characterization (Theorem~\ref{theoremRev}) which is used in the next
section.
A recognizable set  is \emph{recurrent} if its minimal
automaton is strongly connected.

Clearly, a recognizable set
is recurrent if and only if it recognized by a strongly connected
deterministic automaton. Indeed, if $\A$ is a strongly connected deterministic
automaton recognizing $X$, then the minimal automaton of $X$
is also strongly connected.

As a simple and well known example, for any recognizable prefix code $X$, 
the submonoid
$X^*$ generated by $X$ is recurrent.

A recognizable set is \emph{birecurrent} if $X$ and its
reversal $\tilde{X}$ are both recurrent. 

Thus, for example,
the submonoid $X^*$ generated by a bifix code is a birecurrent set.

As a related example, let us consider a reversible automaton,
that is a deterministic  automaton such that every
letter defines an injective map on the set of states.
This class of automata has been considered frequently
(see for example~\cite{Angluin1982},~\cite{Reutenauer1979} or~\cite{Pin1992}).
The reversal of a reversible automaton $\A$ is still
deterministic and thus the set recognized by a strongly connected reversible
automaton is birecurrent.

The two examples of submonoids generated by bifix codes and of
sets recognized by strongly connected reversible automata partially
overlap. Indeed, if $\A=(Q,i,i)$ is a reversible automaton
with a unique terminal state equal to the initial state, the
set recognized by $\A$ is a submonoid generated by a bifix code.
However, there are many examples of submonoids generated by
bifix codes which are not recognizable by a reversible automaton
(see for instance the bifix code of Example~\ref{exampleDegree3bis}).

We will see that there are many other examples of birecurrent sets.

We begin by recalling some basic and well-known facts concerning the
reversal of an automaton.
\subsection{Reversal of an automaton}

Let $\A=(Q,I,T)$ be an automaton. The \emph{reversal}
of $\A$ is the automaton $\tilde{A}=(Q,T,I)$ obtained by reversing
the edges of $\A$ and exchanging $I$ and $T$. Clearly, the reversal of
$\A$ recognizes the reversal $\tilde{X}$ of the set $X$ recognized by $\A$.

For an automaton $\A=(Q,I,T)$, we denote by $\A^\delta$ the \emph{
determinization} of $\A$. Its states are the nonempty sets
$I\cdot w=\{q\in Q\mid i\edge{w}q\mbox{ for some $i\in I$}\}$.
The initial state is $I$ and the set of terminal states
is the family of states $U$ of $\A^\delta$ such that
$U\cap T\ne\emptyset$.

The \emph{deterministic reversal}
of $\A$ is the automaton $\tilde{\A}^\delta$ obtained by determinization
of the reversal of $\A$. We denote by $\tilde{Q}$ the set
of states of $\tilde{\A}^\delta$. Thus $\tilde{Q}$
is the family
of nonempty sets of the form
\begin{displaymath}
w\cdot T=\{q\in Q\mid q\edge{w}t\mbox{ for some $t\in T$}\}.
\end{displaymath}
The set of terminal states
of $\tilde{\A}^\delta$ is $\{U\in \tilde{Q}\mid I\cap U\ne\emptyset\}$. Clearly,
$\tilde{\A}^\delta$ recognizes the reversal of the set
recognized by $\A$.

The following statement is well known (see~\cite{Eilenberg1974} p. 48).
\begin{proposition}\label{propositionJojo}
If $\A$ is a trim deterministic automaton recognizing $X$, then $\tilde{\A}^\delta$
is the minimal automaton of $\tilde{X}$.
\end{proposition}
\begin{proof}
 Since $\A$ is trim, for any word $w$, one has
$w\cdot T\ne\emptyset$ if and only if $Xw^{-1}\ne\emptyset$.
Moreover, for any
$w,w'\in A^*$, one has
\begin{displaymath}
w\cdot T=w'\cdot T\Leftrightarrow Xw^{-1}=Xw'^{-1}
\end{displaymath}
as one may easily verify.
Since the nonempty sets $Xw^{-1}$ are the reversals of the states
of the minimal automaton of $\tilde{X}$, the map $w\cdot T\mapsto \tilde{w}^{-1}\tilde{X}$
is a bijection which identifies $\tilde{\A}^\delta$ with the minimal
automaton of $\tilde{X}$.
\end{proof}

It follows from Proposition~\ref{propositionJojo} that if $\A$
is the minimal automaton of a recognizable set $X$, then
$X$ is birecurrent if and only if $\A$ and  $\tilde{\A}^\delta$
are strongly connected.

\subsection{A characterization of birecurrent sets}\label{sectionRev}
The following statement characterizes birecurrent sets.
One direction appears as \cite[Proposition 5.7]{Perrin2013}.

We say that a set $S\subset Q$ is
\emph{saturated} by a word $w$ if $S$ is a union of classes
of the kernel of $w$. Note that $S$ is saturated by $w$
if and only if $S=w\cdot U$ for some $U\subset Q$.
 
\begin{theorem}\label{theoremRev}
Let $S$ be a recurrent set and let $\A=(Q,i,T)$
be its minimal automaton. Then $S$ is birecurrent
if and only if
$T$ is saturated by a word of minimal nonzero
rank.
\end{theorem}
\begin{proof}
Let us first show that the condition is necessary. Let $v$
be a word of nonzero minimal rank such that $i\cdot v\in T$.
Since $v\cdot T\ne\emptyset$ and since $\tilde{\A}^\delta$
is strongly connected, there is a word $u$ such that
$(uv)\cdot T=T$. Thus $T$ is saturated by
 $uv$.

Conversely, set $\varphi=\varphi_\A$, $M=M_\A$ and let
$J$ be the $0$-minimal ideal of $M$.  Assume that $T$ is saturated by
 a word $x$ of  minimal nonzero rank. Then $\varphi(x)\in J$.
Let $u$ be a word such that
$u\cdot T\ne\emptyset$. We have to show that there
is a word $w$ such that $(wu)\cdot T=T$. This will prove that
$T$ is accessible from $u\cdot T$ in $\tilde{\A}$, which implies that
$\tilde{\A}^\delta$ is strongly connected. 

Since $u\cdot T\ne\emptyset$,
we have $\varphi(ux)\ne 0$. Since $\varphi(x)\in J$,
the left ideal generated by $\varphi(x)$ is $0$-minimal. Thus
 $\varphi(ux)$  is in 
the $\GL$-class of $x$. Let $v$ be a word such that $\varphi(vux)=\varphi(x)$. Since $\varphi((vu)^nx)=\varphi(x)$ for all $n\ge 1$, the idempotent $e$ which is a power of
$\varphi(vu)$ is  such that $e\varphi(x)=\varphi(x)$. Since $T$ is 
saturated by $x$, there is a set $U$ such that $T=x\cdot U$. Then
$eT=e(x\cdot U)=e\varphi(x)U=\varphi(x)U=T$. Finally, let $m$ be such
that $e=\varphi(vu)^m$ and let $w=(vu)^{m-1}v$.
Then $(wu)\cdot T=eT=T$ and the proof is complete.
\end{proof}

Note that if the minimal automaton 
$\A=(Q,i,T)$ of a birecurrent set $S$ is complete and synchronized, then
$\A$ is the trivial automaton with only one state
and $S=A^*$.
This can of course be proved directly, but it follows easily
from the characterization above.
Indeed, in this case, $T$
is saturated by a word of minimal rank
if and only if $T=Q$, which implies the conclusion $S=A^*$.

Theorem~\ref{theoremRev} gives a method
to find birecurrents sets starting with a stongly connected automaton
with an uspecified set of terminal states.
Provided one knows
a word $x$ of minimal nonzero rank, one may choose as set
of terminal states a set saturated by $x$ and finally minimize
the resulting automaton. 

However, this does not give a substantially more efficient
algorithm than the computation of the deterministic reversal.
Indeed, it is shown 
 in~\cite{Ryzhkov2017}
that deciding whether a partial automaton is
strongly connected as well as its deterministic reversal is
PSPACE-complete.

The following examples illustrate Theorem~\ref{theoremRev}. 
\begin{example}\label{exampleBirecurrent1}
The automaton represented in Figure~\ref{figureRev} on the left
and  its deterministic reversal on the right are both strongly connected. 
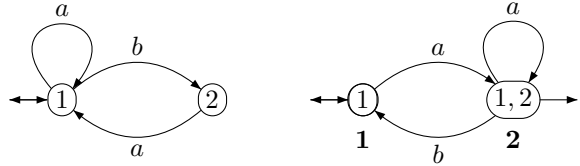
\begin{figure}[hbt]
\centering\gasset{Nadjust=wh}
\begin{picture}(80,20)
\put(0,0){
\begin{picture}(30,20)(0,-5)

\node[Nmarks=if,fangle=180](1)(0,0){$1$}\node(2)(20,0){$2$}

\drawloop(1){$a$}\drawedge[curvedepth=5](1,2){$b$}\drawedge[curvedepth=5](2,1){$a$}
\end{picture}
}
\put(40,0){
\begin{picture}(30,20)(0,-5)

\node[Nmarks=if,fangle=180](1)(0,0){$1$}\node[Nmarks=f,fangle=0](2)(20,0){$1,2$}
\node[ExtNL=y,NLangle=-90,NLdist=2](1)(0,0){$\mathbf{1}$}
\node[ExtNL=y,NLangle=-90,NLdist=2,Nframe=n](3)(20,0){$\mathbf{2}$}
\drawloop(2){$a$}\drawedge[curvedepth=5](1,2){$a$}\drawedge[curvedepth=5](2,1){$b$}
\end{picture}
}
\end{picture}
\caption{A deterministic automaton and its deterministic reversal.}\label{figureRev}
\end{figure}
In agreement with Theorem~\ref{theoremRev},
the set $T=\{1\}$ of terminal states of the first automaton 
is the preimage of $b$, which has rank $1$.

\end{example}

We give a second example with a  minimal rank larger than $1$.
\begin{example}\label{examplePalindrome}
Consider the complete deterministic automaton $\A$ given in Figure~\ref{figureWeakly2} on the left with
its deterministic reversal represented on the right. 
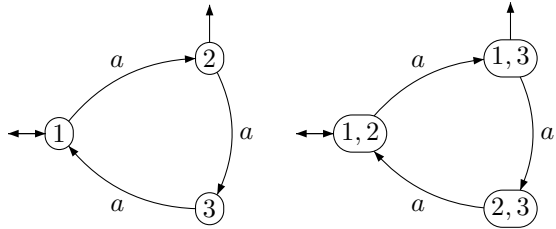
\begin{figure}[hbt]
\centering\gasset{Nadjust=wh}
\begin{picture}(60,25)
\put(0,0){
\begin{picture}(30,20)
\node[Nmarks=if,fangle=180](1)(0,20){$1$}\node[Nmarks=f](2)(20,20){$2$}
\node(3)(20,0){$3$}\node(4)(0,0){$4$}

\drawedge(1,2){$a$}\drawedge[curvedepth=3](1,3){$b$}
\drawedge(2,3){$a,b$}
\drawedge(3,4){$a$}\drawedge[curvedepth=3](3,1){$b$}
\drawedge(4,1){$a,b$}
\end{picture}
}
\put(40,0){
\begin{picture}(30,20)
\node[Nmarks=if,fangle=180](12)(0,20){$1,2$}\node(23)(0,0){$2,3$}
\node(34)(20,0){$3,4$}\node[Nmarks=f,fangle=0](14)(20,20){$1,4$}

\drawedge(23,12){$a,b$}\drawedge[curvedepth=3](12,34){$b$}
\drawedge(34,23){$a$}
\drawedge(14,34){$a,b$}\drawedge[curvedepth=3](34,12){$b$}
\drawedge(12,14){$a$}
\end{picture}
}
\end{picture}
\caption{The automata $\A$ and $\tilde{\A}^\delta$ in Example~\ref{examplePalindrome}.}\label{figureWeakly2}
\end{figure}
The $0$-minimal ideal of the monoid $\varphi_\A(A^*)$ is represented in Figure~\ref{figureMinimalIdeal2}. We can check that the set $T=\{1,2\}$ is a class
of the kernel of $b$, which is a word of minimal rank equal to $2$.
\begin{figure}[hbt]

\begin{displaymath}
\def\rb{\hspace{2pt}\raisebox{0.8ex}{*}}\def\vh{\vphantom{\biggl(}}
    \begin{array}%
    {r|@{}l@{}c|@{}l@{}c|}%
    \multicolumn{1}{r}{}&\multicolumn{2}{c}{1/3}&\multicolumn{2}{c}{2/4}\\
    \cline{2-5}
    1,2/3,4& \vh\rb &b &\vh\rb  &ba \\
    \cline{2-5}
    1,4/2,3&\vh\rb &ab &\vh\rb& aba \\
    \cline{2-5}
    \end{array}
\end{displaymath}
\caption{The $0$-minimal ideal of $\varphi_\A(A^*)$ in Example~\ref{examplePalindrome}.}\label{figureMinimalIdeal2}
\end{figure}
\end{example}
In the last example, we show a simple case in which $T$ is a union
of several classes of the kernel of a word of minimal nonzero rank.
\begin{example}\label{example3/2}
Let $\A$ be the deterministic automaton on $Q=\{1,2,3\}$ and with $A=\{a\}$ where
$a$ is the circular permutation $(123)$. Choosing $i=1$
and $T=\{1,2\}$, we obtain
for $\tilde{\A}^\delta$ the automaton on the set of states $\{\{1,2\},\{1,3\},\{2,3\}\}$
on which $a$ acts again as a circular permutation. Thus $\tilde{\A}^\delta$ is 
strongly connected (see Figure~\ref{figureGroupAutomaton}).
Actually, this holds for
any strongly connected group automaton.

\begin{figure}[hbt]
\centering\gasset{Nadjust=wh}
\begin{picture}(60,30)(0,-5)
\put(0,0){
\begin{picture}(30,20)
\node[Nmarks=if,fangle=180](1)(0,10){$1$}
\node[Nmarks=f,fangle=90](2)(20,20){$2$}
\node(3)(20,0){$3$}

\drawedge[curvedepth=3](1,2){$a$}
\drawedge[curvedepth=3](2,3){$a$}
\drawedge[curvedepth=3](3,1){$a$}

\end{picture}
}
\put(40,0){
\begin{picture}(30,20)
\node[Nmarks=if,fangle=180](1)(0,10){$1,2$}
\node[Nmarks=f,fangle=90](2)(20,20){$1,3$}
\node(3)(20,0){$2,3$}

\drawedge[curvedepth=3](1,2){$a$}
\drawedge[curvedepth=3](2,3){$a$}
\drawedge[curvedepth=3](3,1){$a$}

\end{picture}
}
\end{picture}
\caption{The automata $\A$ and $\tilde{\A}^\delta$ in Example~\ref{example3/2}.}\label{figureGroupAutomaton}
\end{figure}
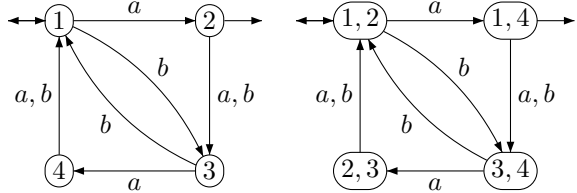
\end{example}
\subsection{Birecurrent sets of finite type}\label{sectionFiniteType}
In this section, we define birecurrent sets of finite type and 
in the next section, we prove
a statement which allows to build an infinite family of such sets
(Theorem~\ref{theoremDP}).
We first prove the following elementary statement concerning recurrent sets.
\begin{proposition}\label{propositionRecurrent}
A set $S\subset A^*$ is recurrent if and only if there is a recognizable
prefix code $X$
and a recognizable set $P$  of proper prefixes of $X$ such that
$S=X^*P$. 
\end{proposition}
\begin{proof}
Assume first that $S$ is recurrent. Let $\A=(Q,i,T)$ be the minimal
automaton of $S$. Let $X$ be the prefix code generating the submonoid recognized
by the automaton $(Q,i,i)$ and let $P$ be the set of proper prefixes $p$ of $X$
such that $i\cdot p\in T$. Then $S=X^*P$.

Conversely, let $\A=(Q,\varepsilon,\varepsilon)$ be the literal automaton of $X^*$. Recall
from Section~\ref{sectionPreliminaries} that $Q$ is the set
of proper prefixes of $X$ and that $\A$ recognizes $X^*$.
Then $S$ is recognized by the automaton $(Q,\varepsilon,P)$ which is a strongly connected
automaton. Moreover, since $X$ and $P$ are recognizable, $S$
is recognizable. Thus $S$ is recurrent.
\end{proof}
A pair $(X,P)$ as above is called a \emph{decomposition}
of the recurrent set $S$.
 The prefix code $X$ obtained as above
using the minimal automaton of $S$ is called the \emph{left root}
of $S$. 

Note that a recurrent set has in general several decompositions.
However, the left root $X$ of $S$ is such that 
for any such pair $(X',P')$, we have $X'\subset X^*$,
 that is the submonoid $X^*$ is
maximal for this property. Indeed, let $(Q',\varepsilon,\varepsilon)$
be the literal automaton of $X'^*$. Then $\A'=(Q',\varepsilon,P')$ recognizes
$S$ and thus there is a reduction from $\A'$ onto the minimal automaton
$(Q,i,P)$ of $S$ which sends $\varepsilon$ to $i$. The image by this reduction of a path in $\A'$
from $\varepsilon$ to $\varepsilon$ is a path in $\A$
from $i$ to $i$ and thus $X'\subset X^*$.

\begin{proposition}\label{propositionRoot}
The left root of a recurrent set $S$ is finite if and only if $S$ has a decomposition
$(X,P)$ with $X$ finite.
\end{proposition}
\begin{proof}
Let $(X,P)$ and $(X',P')$ be two decompositions of a recurrent set $S$
with $X$  the left root of $S$. Assume that $X'$ is finite. 

Any word of $X$ is a prefix of a word in $X'$. Indeed, let $x\in X$
and let $p$ be some element of $P$. Then $xp=x'p'$ with $x'\in X'^*$ and
$p'\in P'$. Thus $x$ is a prefix of a word in $X'^*$. Since $X'\subset X^*$,
this implies that $x$ is a prefix of a word in $X'$. This shows that
the length of $x$ is bounded by the maximal length of the words in $X'$.
Therefore $X$ is finite.

The other implication is clear.
\end{proof}

For a birecurrent set $S$, we have
\begin{equation}
S=X^*P=QY^*
\end{equation}
where $\tilde{Y}$ is the left root of $\tilde{S}$
and $Q$ is a set of proper suffixes of $Y$. The suffix code $Y$ 
is called the
\emph{right root} of $S$. When $S$ is the submonoid generated by a bifix code $X$, then $X$ is the left root of $S$ and $\tilde{X}$ is its right root.

We say that a birecurrent set $S$ is of \emph{finite type} if its 
left and right roots are finite. 

For example, for every finite bifix code $X$, the set
$S=X^*$ is a birecurrent set of finite type.

On the contrary, the birecurrent set of Example~\ref{exampleBirecurrent1}
is not of finite type (see Example~\ref{example13}).

We give two examples of birecurrent sets of finite type. The first
one is Example~\ref{examplePalindrome}.
Recall from Section~\ref{sectionPreliminaries}
that $\underline{Y}$  denotes
the characteristic series of a set $Y\subset A^*$. Note that when $(X,P)$ is a decomposition of
a recurrent set $S$, we have $\underline{S}=\underline{X}^*\underline{P}$.
Indeed, since $X$ is a prefix code, we have $\underline{M}=(\underline{X})^*$
for the submonoid $M=X^*$ and $\underline{S}=\underline{M}\ \underline{P}$
because the product $(M,P)$ is unambiguous (see~\cite{BerstelPerrinReutenauer2009}).
\begin{example}\label{examplePalindrome2}
The birecurrent set $S$ of Example~\ref{examplePalindrome}
is of finite type. Indeed, its left root is the finite prefix code 
$X=Z^2$ with $Z=aA\cup b$ and its right root is the finite suffix code 
$\tilde{X}$. One has $S=X^*P=P\tilde{X}^*$ with $P=\{\varepsilon,a\}$.
One has actually $S=\tilde{S}$. This can be checked by comparing
the two automata of Figure~\ref{figureWeakly2}
or directly by observing that $\underline{Z}=a^2+(1+a)b$
implies that $\underline{Z}(1+a)=(1+a)\underline{\tilde{Z}}$,
whence $X^*P=P\tilde{X}^*$ since $\underline{P}=1+a$.
\end{example}

The second example is Example 3.6.13 in~\cite{BerstelPerrinReutenauer2009}.
\begin{example}\label{exampleDegree3}
Let $\A=(Q,i,T)$ be the automaton given by Table~\ref{tableAutomaton}
with $i=1$ and $T=\{1,6\}$.
\begin{table}[hbt]
\begin{displaymath}
\begin{array}{c| c c c c c c c c c}
 & 1 & 2 & 3 & 4 & 5 & 6 & 7 & 8 & 9\\ \hline
a& 2 & 3 & 1 & 1 & 3 & 8 & 9 & 3 & 1\\ \hline
b& 4 & 6 & 7 & 5 & 1 & 4 & 1 & 5 & 1
\end{array}
\end{displaymath}
\caption{The transitions of the automaton $\A$ in Example~\ref{exampleDegree3}.}\label{tableAutomaton}
\end{table}
The minimal rank is equal to $3$ as one may check by computing
the minimal images which are $\{1,2,3\}$, $\{4,6,7\}$,
$\{1,4,5\}$ and $\{1,8,9\}$. The set $\{1,6\}$ is a class
of the kernel of $a^2$ which has image $\{1,2,3\}$ and
thus minimal rank. The  automaton $\tilde{\A}^\delta$
is thus strongly connected by Theorem~\ref{theoremRev} and the
set $S$ recognized by $\A$ is birecurrent. The left root of  $S$
 is the finite maximal prefix code
represented in Figure~\ref{figureLeftRoot} as the set of leaves
of a binary tree. In this figure,
as in the following ones, we represent words on the alphabet $\{a,b\}$
by nodes of a binary tree. The edges going up are labeled $a$,
those going down are labeled $b$.
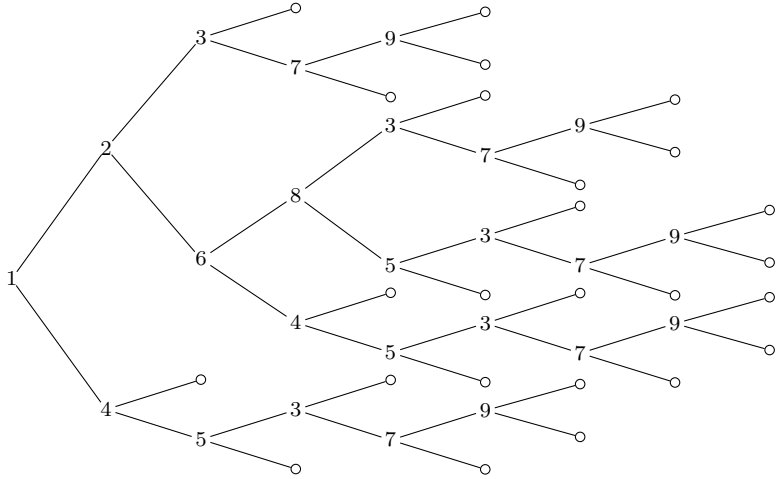
\begin{figure}[hbt]
  \centering
  \psset{style=simpletreestyle}
  \footnotesize
\begin{pspicture}(0,0)(15,9)
  \renewcommand\treesepvalue{1}
  \psset{treenodesize=0.0} 
  \psset{levelsep=1.8} 
  \def\Sero{\TC}
  \def\SNine{\MTu9{\Sero\Sero}}
  \def\SSeven{\MTu7{\SNine\Sero}}
  \def\SThree{\MTu3{\Sero\SSeven}}
  \def\SFive{\MTu5{\SThree\Sero}}
  \def\SFour{\MTu4{\Sero\SFive}}
  \def\SEight{\MTu[0.4]8{\SThree\SFive}}
  \def\SSix{\MTu[0.4]6{\SEight\SFour}}
  \def\STwo{\MTu[0.4]2{\SThree\SSix}}
  \def\SOne{\MTu[0.4]1{\STwo\SFour}}
  \rput[lb](0,0){\SOne}
\end{pspicture}
\caption{The left root of the set $S$ i Example~\ref{exampleDegree3}.}
\label{figureLeftRoot}
\end{figure}

The deterministic reversal of $\A$ 
has the set of states represented in Table~\ref{tableStatesTilde}.
\begin{table}[hbt]
\begin{displaymath}
\begin{array}{|c|c|c|c|c|c|c|c|c|}\hline
\mathbf{1}  & \mathbf{2}   & \mathbf{3}     & \mathbf{4}     & \mathbf{5} & \mathbf{6}   & \mathbf{7}     & \mathbf{8}     & \mathbf{9}\\ \hline
1,6&3,4,9&2,5,7,9&2,5,7,8&1,7&3,4,8&3,5,7,9&2,5,6,8&2,4,8\\ \hline
\end{array}
\end{displaymath}
\caption{The states of $\tilde{\A}^\delta$ in Example~\ref{exampleDegree3}.}\label{tableStatesTilde}
\end{table}

The transitions of $\tilde{\A}^\delta$ are given in Table~\ref{tableTilde}.
\begin{table}[hbt]
\begin{displaymath}
\begin{array}{c|ccccccccc}
  & \mathbf{1}  & \mathbf{2}   & \mathbf{3}     & \mathbf{4}     & \mathbf{5} & \mathbf{6}   & \mathbf{7}     & \mathbf{8}     & \mathbf{9}\\ \hline
a & \mathbf{2}  & \mathbf{4}   & \mathbf{5}     & \mathbf{1}     & \mathbf{2} & \mathbf{8}   & \mathbf{4}     & \mathbf{1}     & \mathbf{1}\\ \hline
b & \mathbf{3}  &\mathbf{1}   & \mathbf{6}     & \mathbf{6}     & \mathbf{7} & \mathbf{1}   & \mathbf{6}     & \mathbf{9}     & \mathbf{1}
\end{array}
\end{displaymath}
\caption{The transitions of the automaton $\tilde{\A}^\delta$
in Example~\ref{exampleDegree3}.}\label{tableTilde}
\end{table}
The reversal $\tilde{Y}$ of the right root of $S$ is represented
in Figure~\ref{figureRightRoot}.
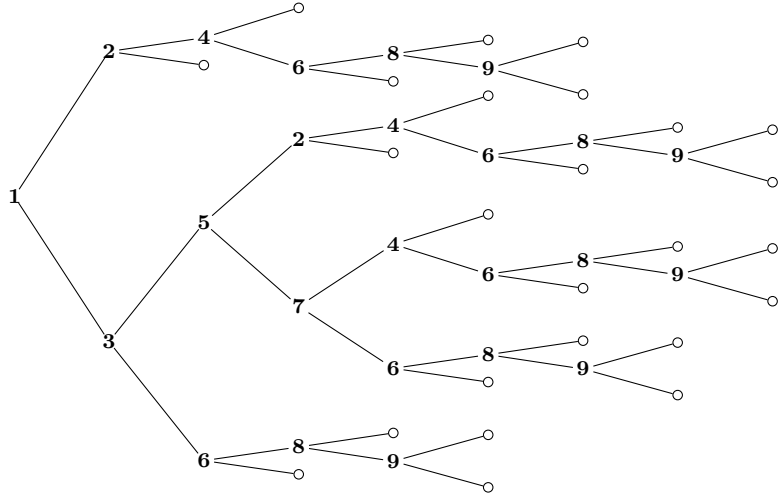
\begin{figure}[hbt]
  \centering
  \psset{style=simpletreestyle}
  \footnotesize
\begin{pspicture}(0,0)(15,9)
  \renewcommand\treesepvalue{1}
  \psset{treenodesize=0.0} 
  \psset{levelsep=1.8} 
  \def\Sero{\TC}
  \def\SNine{\MTu{\mathbf{9}}{\Sero\Sero}}
  \def\SSeven{\MTu{\mathbf{7}}{\SFour\SSix}}
  \def\SThree{\MTu{\mathbf{3}}{\SFive\SSix}}
  \def\SFive{\MTu{\mathbf{5}}{\STwo\SSeven}}
  \def\SFour{\MTu{\mathbf{4}}{\Sero\SSix}}
  \def\SEight{\MTu[0.4]{\mathbf{8}}{\Sero\SNine}}
  \def\SSix{\MTu[0.4]{\mathbf{6}}{\SEight\Sero}}
  \def\STwo{\MTu[0.4]{\mathbf{2}}{\SFour\Sero}}
  \def\SOne{\MTu[0.4]{\mathbf{1}}{\STwo\SThree}}
  \rput[lb](0,0){\SOne}
\end{pspicture}
\caption{The reversal of the right root of the set $S$ in Example~\ref{exampleDegree3}.}
\label{figureRightRoot}
\end{figure}

Since $\tilde{Y}$ is finite, $S$ is a  birecurrent set of finite type.
\end{example}
We give now two examples of  birecurrent sets such that the left 
root is finite but the right root is infinite and thus
which are not of finite type.

\begin{example}\label{example13}
Consider again the set $S$ recognized by the automaton of
Figure~\ref{figureRev} on the left with its deterministic reversal on the right
(Example~\ref{exampleBirecurrent1}). The left root of $S$ is
$\{a,ba\}$ and thus it is finite. However, the
right root of $S$ is $ba^+$ and it is infinite. Note that the two factorizations
of $S$ correspond to the identity
\begin{displaymath}
\{a,ba\}^*=a^*(ba^+)^*
\end{displaymath}
which is itself a particular case of the \emph{sumstar}
identity $(x+y)^*=x^*(yx^*)^*$.
\end{example}
In the second example, the birecurrent set $S$ is dense.
\begin{example}\label{example14}
Consider the finite maximal prefix code $X=\{aa,aba,abb,b\}^2$. The minimal
automaton  of $X^*$ is given by its transitions in 
Table~\ref{tableNotFiniteType} on the left with  $1$ as initial
and terminal state.

\begin{table}[hbt]
\begin{displaymath}
\begin{array}{c|cccccc}
 &1&2&3&4&5&6\\ \hline
a&2&3&4&1&3&1\\
b&3&5&1&6&3&1
\end{array}
\quad
\begin{array}{c|cccccc}
 &\mathbf{1}&\mathbf{2}&\mathbf{3}&\mathbf{4}&\mathbf{5}&\mathbf{6}\\ \hline
a&\mathbf{2}&\mathbf{4}&\mathbf{1}&\mathbf{6}&\mathbf{4}&\mathbf{1}\\
b&\mathbf{3}&\mathbf{4}&\mathbf{5}&\mathbf{5}&\mathbf{3}&\mathbf{1}
\end{array}
\end{displaymath}
\caption{The transitions of $\A$ and $\tilde{\A}^\delta$
in Example~\ref{example14}.}\label{tableNotFiniteType}
\end{table}
The word $ab^2$ has minimal rank $2$ and its kernel is $\{1,2,5\},\{3,4,6\}$.
Keeping the same initial state and taking $T=\{1,2,5\}$, we obtain a 
deterministic automaton $\A$
recognizing the set $S=X^*P$ with $P=\{\varepsilon,a,ab\}$.
The set
of states of $\tilde{\A}^\delta$ is given in Table~\ref{tableStatesMirror}.
\begin{table}[hbt]
\begin{displaymath}
\begin{array}{|c|c|c|c|c|c|}\hline
\mathbf{1}&\mathbf{2}&\mathbf{3}&\mathbf{4}&\mathbf{5}&\mathbf{6}\\\hline
1,2,5&1,4,6&2,3,6&3,4,6&1,4,5&2,3,5\\\hline
\end{array}
\end{displaymath}
\caption{The states of $\tilde{\A}^\delta$ in Example~\ref{example14}.}\label{tableStatesMirror}
\end{table}
Its transitions are given in Table~\ref{tableNotFiniteType} on the right. The 
right
root of $S$ is the maximal suffix code $\tilde{Y}$
where $Y^*$ is the submonoid recognized by $\tilde{\A}^\delta$ with
$\mathbf{1}$ as initial and terminal state. Since, in $\tilde{\A}^\delta$, there
is a loop $\mathbf{3}\stackrel{b}{\rightarrow}\mathbf{5}\stackrel{b}{\rightarrow}\mathbf{3}$,
the code $Y$ is infinite.
\end{example}
\subsection{A construction of birecurrent sets of finite type}\label{sectionConstruction}
Examples~\ref{examplePalindrome2} and \ref{exampleDegree3}
are particular cases of a general construction that
we now describe. 
Let $Z\subset A^+$ be a set of
words.
A word $x\in Z$ is said to be a \emph{pure square} for $Z$ if $x=w^2$
for some $w\in A^+$ and if $Z\cap wA^*\cap A^*w=\{x\}$.
 The following result is~ \cite[Exercise 14.1.9 (a)]{BerstelPerrinReutenauer2009}.
\begin{proposition}\label{propDelta}
If $Z$ is a finite maximal prefix code and if $w^2$ is
a pure square for $Z$, then, denoting $G=Zw^{-1}$ and
$D=w^{-1}Z$, the expression
\begin{equation}
(1+w)(\underline{Z}-1+(\underline{G}-1)w(\underline{D}-1))+1\label{eqDelta}
\end{equation}
is the characteristic polynomial of a finite maximal prefix code.
Moreover, the expression
\begin{equation}
(\underline{Z}-1+(\underline{G}-1)w(\underline{D}-1))(1+w)+1\label{eqDelta2}
\end{equation}
is the characteristic polynomial of a finite maximal code.
\end{proposition}
The prefix code defined by
Equation~\eqref{eqDelta} is denoted $\delta_w(Z)$.
The maximal code defined by Equation~\eqref{eqDelta2} is
denoted $\gamma_w(Z)$. If $Z$ is bifix, then $\gamma_w(Z)$
is a suffix code since $\gamma_w(Z)$ is the reversal
of $\delta_{\tilde{w}}(\tilde{Z})$.

\begin{theorem}\label{theoremDP}
Let $Z\subset A^+$ be a finite maximal bifix code 
and let $x=w^2$ be a pure square for $Z$. 
The set $S=\delta_w(Z)^*\{\varepsilon,w\}$ is a dense birecurrent set of finite type.
\end{theorem}
\begin{proof} Set $X=\delta_w(Z)$ and $Y=\gamma_w(Z)$. Observe that, by defintion of $\delta_w(Z)$ and $\gamma_w(Z)$,
\begin{eqnarray*}
(1+w)(\underline{Y}-1)&=&(1+w)(\underline{Z}-1+(\underline{G}-1)w(\underline{D}-1))(1+w)\\
&=&(\underline{X}-1)(1+w).
\end{eqnarray*}
 Thus, multiplying on the left by $\underline{X}^*$ and on the
right by $\underline{Y}^*$, we obtain $\underline{X}^*(1+w)=(1+w)\underline{Y}^*$. This shows that
$S=X^*\{\varepsilon,w\}=\{\varepsilon,w\}Y^*$.

 Since $S=X^*\{\varepsilon,w\}$, the set $S$ is recurrent by Proposition~\ref{propositionRecurrent}
and its left root is finite by Proposition~\ref{propositionRoot}.

Similarly, since $S=\{\varepsilon,w\}Y^*$, the set $S$  is birecurrent.
Since  $Y$ finite, the right root of $S$
is also finite. Thus $S$ is birecurrent of finite type.

Since $X$ is a maximal prefix code, the submonoid $X^*$ is right dense.
Thus $S$ is right dense. 
\end{proof}
Observe that if $R$ is the set of proper prefixes of $Z$ and
$U$ the set of proper prefixes of $D$, the characteristic polynomial
of the set $P$ of proper
prefixes of $X=\delta_w(Z)$ is
\begin{equation}
\underline{P}=(1+w)(\underline{R}+(\underline{G}-1)w\underline{U}).\label{eqPrefDelta}
\end{equation}
Indeed, since
$Z,X$ and $D$ are maximal prefix codes, we have
$\underline{Z}-1=\underline{R}(\underline{A}-1)$,
 $\underline{X}-1=\underline{P}(\underline{A}-1)$ and
 $\underline{D}-1=\underline{U}(\underline{A}-1)$.
Thus Equations~\eqref{eqDelta} and \eqref{eqPrefDelta}
are equivalent.

In the next two examples, we show that Examples~\ref{examplePalindrome2}
and \ref{exampleDegree3} can be obtained by the contruction 
decribed in Theorem~\ref{theoremDP}.
\begin{example}\label{examplePalindrome3}
Let $S$ be the birecurrent set of Example~\ref{examplePalindrome}.
We have seen in Example~\ref{examplePalindrome2} that
$S=X^*P$ with $X=(aA\cup b)^2$ and $P=\{\varepsilon,a\}$. We obtain $X$
as in Theorem~\ref{theoremDP} using $Z=A^2$ and $x=a^2$. Indeed,
we have $G=D=A$ and thus Equation~\eqref{eqPrefDelta}
gives
\begin{displaymath}
(1+a)(1+a+b+(a+b-1)a)=(1+a)(1+Aa+b)=(1+aA+b)(1+a)
\end{displaymath}
which is the characteristic polynomial of the set of proper prefixes of $X$.
\end{example}
\begin{example}\label{exampleDegree3bis}
Let $S$ be the birecurrent set of Example~\ref{exampleDegree3}.
We start from the finite maximal bifix code
\begin{displaymath}
Z=\{aaa,aab,abaa,abab,abb,ba,bbaa,bbab,bbb\}
\end{displaymath}
which is represented in Figure~\ref{fig3_03} on the left
with its reversal on the right.
\begin{figure}[hbt]
  \centering
  \psset{style=simpletreestyle}\psset{radius=0.1} 
  \renewcommand\treesepvalue{0.9}
    \small
    \def\Ttwo{\MTp{\TC\TC}}
    \def\Tleft{\MTp{\TC\Ttwo}}
    \def\Tright{\MTp{\Ttwo\TC}}
\begin{pspicture}(0,0)(6.5,5)
    \rput[l](0,2){%
      \MTp[0.4]{\MTp[0.4]{\Ttwo\Tright}\MTp{\TC\Tright}}
    }
  \end{pspicture}\qquad
  \begin{pspicture}(0,0)(6.5,4)
    \rput[l](0,2){%
      \MTp[0.4]{\MTp{\Tleft\TC}\MTp[0.4]{\Tleft\Ttwo}}
    }
  \end{pspicture}\qquad\qquad
  
  \caption{The bifix code $Z$  and  its
    reversal $\widetilde{Z}$ in Example~\ref{exampleDegree3bis}.}
  \label{fig3_03}
\end{figure}
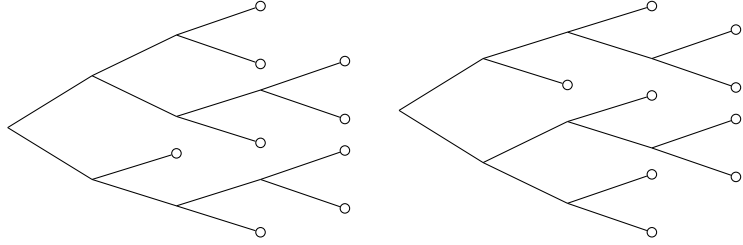

The word $x=w^2$ with $w=ab$ is a pure square for $Z$. We have
$G=\{a,ab,bb\}$ and $D=\{aa,ab,b\}$. 

The set $R$
of proper prefixes of $Z$ is $R=\{\varepsilon,a,b,aa,ab,bb,aba,bba\}$
and the set $U$ of proper prefixes of $D$ is $U=\{\varepsilon,a\}$.

The set $R$ is represented in Figure~\ref{figureR} on the left,
the set $R\setminus wU$ is represented 
in the middle, and the
set $Q=(R\setminus wU)\cup GwU$ on the right (the white nodes being not
in the set).
 
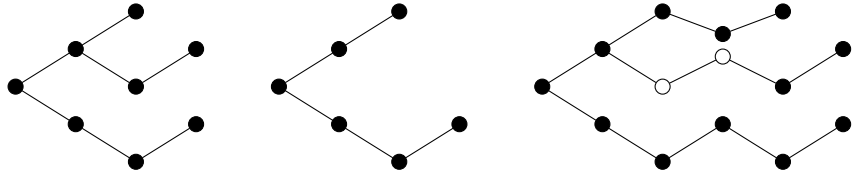
\begin{figure}[hbt]
\centering
\gasset{AHnb=0,Nadjust=wh,fillcolor=black}
\begin{picture}(100,25)
\put(0,0){
\begin{picture}(30,25)
\node(1)(0,15){}
\node(a)(8,20){}\node(b)(8,10){}
\node(aa)(16,25){}\node(ab)(16,15){}\node(bb)(16,5){}
\node(aba)(24,20){}\node(bba)(24,10){}

\drawedge(1,a){}\drawedge(1,b){}
\drawedge(a,aa){}\drawedge(a,ab){}\drawedge(ab,aba){}
\drawedge(b,bb){}\drawedge(bb,bba){}
\end{picture}
}
\put(35,0){

\begin{picture}(30,25)
\node(1)(0,15){}
\node(a)(8,20){}\node(b)(8,10){}
\node(aa)(16,25){}\node(bb)(16,5){}
\node(bba)(24,10){}

\drawedge(1,a){}\drawedge(1,b){}
\drawedge(a,aa){}
\drawedge(b,bb){}\drawedge(bb,bba){}
\end{picture}
}
\put(70,0){
\begin{picture}(30,25)
\node(1)(0,15){}
\node(a)(8,20){}\node(b)(8,10){}
\node(aa)(16,25){}
\node(aab)(24,22){}\node(aaba)(32,25){}
\node[fillcolor=white](ab)(16,15){}
\node[fillcolor=white](aba)(24,19){}
\node(abab)(32,15){}\node(ababa)(40,20){}
\node(bb)(16,5){}\node(bba)(24,10){}
\node(bbab)(32,5){}\node(bbaba)(40,10){}

\drawedge(1,a){}\drawedge(1,b){}
\drawedge(a,aa){}\drawedge(a,ab){}
\drawedge(aa,aab){}\drawedge(ab,aba){}
\drawedge(aab,aaba){}\drawedge(aba,abab){}
\drawedge(abab,ababa){}
\drawedge(b,bb){}\drawedge(bb,bba){}
\drawedge(bba,bbab){}\drawedge(bbab,bbaba){}
\end{picture}
}
\end{picture}
\caption{The sets $R$, $R\setminus wU$ and $Q=(R\setminus wU)\cup GwU$
in Example~\ref{exampleDegree3bis}.}\label{figureR}
\end{figure}

Finally, the set $P=\{\varepsilon,w\}Q$ is represented in Figure~\ref{figureQ}
(with the nodes of $Q$ in black and those of $wQ$ in red).
By Equation~\eqref{eqPrefDelta}, it is the set of proper prefixes of the maximal prefix code $X=\delta_w(Z)$
represented in Figure~\ref{figureLeftRoot}.
\begin{figure}[hbt]
\centering
\gasset{AHnb=0,Nadjust=wh}
\begin{picture}(80,60)
\put(0,0){
\gasset{fillcolor=black}
\begin{picture}(60,60)
\node(1)(0,30){}
\node(a)(10,40){}\node(b)(10,20){}
\node(aa)(20,50){}
\node(aab)(30,45){}\node(aaba)(40,55){}
\node(ab)(20,30){}
\node(aba)(30,40){}
\node(abab)(40,30){}\node(ababa)(50,40){}
\node(bb)(20,10){}\node(bba)(30,15){}
\node(bbab)(40,5){}\node(bbaba)(50,10){}

\drawedge(1,a){}\drawedge(1,b){}
\drawedge(a,aa){}\drawedge(a,ab){}
\drawedge(aa,aab){}\drawedge(ab,aba){}
\drawedge(aab,aaba){}\drawedge(aba,abab){}
\drawedge(abab,ababa){}
\drawedge(b,bb){}\drawedge(bb,bba){}
\drawedge(bba,bbab){}\drawedge(bbab,bbaba){}
\end{picture}
}
\put(20,0){
\gasset{fillcolor=red}
\begin{picture}(60,60)
\node(1)(0,30){}
\node(a)(10,40){}\node(b)(10,20){}
\node(aa)(20,50){}
\node(aab)(30,45){}\node(aaba)(40,55){}
\node[fillcolor=black](ab)(20,30){}
\node[fillcolor=black](aba)(30,40){}
\node(abab)(40,30){}\node(ababa)(50,40){}
\node(bb)(20,10){}\node(bba)(30,15){}
\node(bbab)(40,10){}\node(bbaba)(50,15){}

\drawedge(1,a){}\drawedge(1,b){}
\drawedge(a,aa){}\drawedge(a,ab){}
\drawedge(aa,aab){}\drawedge(ab,aba){}
\drawedge(aab,aaba){}\drawedge(aba,abab){}
\drawedge(abab,ababa){}
\drawedge(b,bb){}\drawedge(bb,bba){}
\drawedge(bba,bbab){}\drawedge(bbab,bbaba){}
\end{picture}
}
\end{picture}

\caption{The set $P=\{\varepsilon,w\}Q$ in Example~\ref{exampleDegree3bis}.}\label{figureQ}
\end{figure}
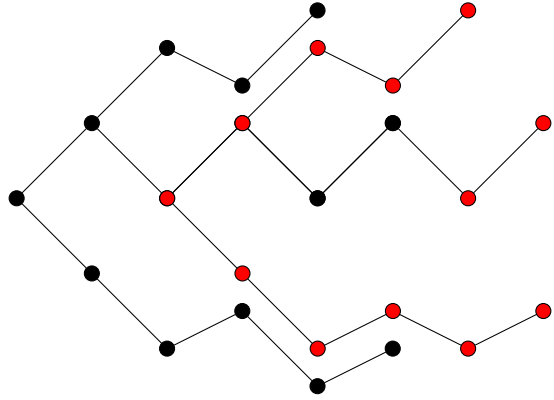
\end{example}

Theorem~\ref{theoremDP} gives an infinite family of examples
of birecurrent sets of finite type which are not submonoids
generated by a bifix code. Indeed, for any even $n\ge 1$, the word
$a^n$ is a pure square for the bifix code $A^n$.
\subsection{Multiple factorizations of noncommutative polynomials}
\label{sectionMultipleFactorizations}
Let $S$ be a dense birecurrent set of finite type. Then $S=X^*P=QY^*$
where $X$ (resp. $Y$) is a finite  maximal prefix code (resp. a finite 
maximal suffix code).
Since all products are non ambiguous, we have the equality
\begin{equation}
(1-\u(X))\u(Q)=\u(P)(1-\u(Y)).\label{eqFact}
\end{equation}
Let $J$ be the set of proper prefixes of $X$ and let $K$ be the
set of proper suffixes of $Y$. Then
\begin{equation}
1-\u(X)=\u(J)(1-\u(A)),\quad 1-\u(Y)=(1-\u(A))\u(K).\label{eqPrefSuf}
\end{equation}
Combining Equations~\eqref{eqFact} and ~\eqref{eqPrefSuf}, we obtain
\begin{equation}
\u(J)(1-\u(A))\u(Q)=\u(P)(1-\u(A))\u(K).\label{eqFact2}
\end{equation}

We conjecture that, with the above notation
 there exist sets $M,N\subset A^*$ such that
\begin{equation}
\u(J)=\u(P)\u(M),\quad \u(K)=\u(N)\u(Q),\quad \u(M)(1-\u(A))=(1-\u(A))\u(N)
\label{eqConjecture}
\end{equation}
This is true when $S$ is generated by a maximal bifix code
and when it is obtained by the construction of Section~\ref{sectionConstruction}. Indeed, in this case, we have $\u(P)=\u(Q)=1+w$ and by Equation~\eqref{eqPrefDelta}
\begin{displaymath}
\u(J)=(1+w)(\u(P_Z)+(\u(G)-1)w\u(P_D))
\end{displaymath}
and 
\begin{displaymath}
\u(K)=(1+w)(\u(S_Z)+(\u(S_G)-1)w(\u(D)-1)).
\end{displaymath}
where $P_Z,P_D$ denote the set of proper prefixes of $Z,D$
and $S_Z,S_G$ denote the set of proper suffixes of $Z,G$.
Thus $\u(J)=\u(P)\u(M)$ with $\u(M)=\u(P_Z)+(\u(G)-1)w\u(P_D)$
and $\u(K)=\u(N)\u(Q)$ with $\u(N)=\u(S_Z)+(\u(S_G)-1)w(\u(D)-1)$.
Moreover $\u(M)(\u(A)-1)=\u(Z)-1+(\u(G)-1)w(\u(D)-1)=(1-\u(A))\u(N)$.

A weak form of this conjecture is proved in~\cite[Theorem 3.1]{BuyereDeFelice1992}.
\subsection{Degree of a dense birecurrent set}\label{sectionDegree}
Let $S$ be a dense birecurrent set. We define its
\emph{degree} as the degree of its minimal automaton. Recall
from Section~\ref{sectionPreliminaries} that the
degree $d(X)$ of a prefix code $X$ is, by definition, the degree of the
minimal automaton of $X^*$.  Thus the
degree of a dense birecurrent set is the degree of its left  root
(this is true even if the minimal automaton
of $X^*$ is a quotient of the minimal automaton of $S$, see~\cite[Proposition 9.5.2]{BerstelPerrinReutenauer2009}).

Thus, when $S=X^*$ with $X$ a maximal bifix code, the degree of $S$
is the degree of $X$.

Note that the degree of the right root may be different from the
degree of the left root, and thus of the degree of the reversal,
 as shown in the following example.

\begin{example}\label{exampleLeftRightDegrees}
Let $S$ be the dense birecurrent set recognized by the automaton $\A$
represented in Figure~\ref{figureS4} on the left. The automaton
$\tilde{\A}^\delta$ is represented on the right.
\begin{figure}[hbt]
\gasset{Nadjust=wh}\centering
\begin{picture}(90,35)
\put(0,0){
\begin{picture}(20,28)(0,-5)
\node[Nmarks=if,fangle=180](1)(0,20){$1$}\node[Nmarks=f](2)(20,20){$2$}
\node(3)(20,0){$3$}\node(4)(0,0){$4$}

\drawedge[curvedepth=3](1,2){$a,b$}\drawedge[curvedepth=3](2,1){$b$}
\drawedge[curvedepth=3](2,3){$a$}\drawedge[curvedepth=3](3,4){$a$}
\drawedge[curvedepth=3](4,1){$a$}
\drawloop[loopangle=0](3){$b$}\drawloop[loopangle=180](4){$b$}
\end{picture}
}
\put(50,0){
\begin{picture}(20,25)
\node[Nmarks=if,fangle=180](1)(0,30){$1,2$}\node(2)(40,30){$2,3$}
\node(3)(40,0){$3,4$}\node[Nmarks=f,fangle=180](4)(0,0){$1,4$}
\node(24)(15,10){$2,4$}\node[Nmarks=f,fangle=-45](13)(25,20){$1,3$}

\drawedge[curvedepth=-3](2,1){$a$}
\drawedge[curvedepth=-3](3,2){$a$}\drawedge[curvedepth=-3](4,3){$a$}
\drawedge[curvedepth=-3](1,4){$a$}
\drawedge[curvedepth=3](4,24){$b$}\drawedge[curvedepth=3](24,4){$b$}
\drawedge[curvedepth=3](24,13){$a$}\drawedge[curvedepth=3](13,24){$a$}
\drawedge[curvedepth=3](2,13){$b$}\drawedge[curvedepth=3](13,2){$b$}

\drawloop[loopangle=180](1){$b$}\drawloop[loopangle=0](3){$b$}
\end{picture}
}
\end{picture}
\caption{The automata $\A$ and $\tilde{\A}^\delta$
in Example~\ref{exampleLeftRightDegrees}.}\label{figureS4}
\end{figure}
The degree of $S$ is $4$ since the transition monoid of $\A$ is the symmetric
group $\mathfrak{S}_4$ on $4$ elements. But the degree of the right root of $S$ is $6$
since the transition monoid of the automaton $\tilde{\A}^\delta$ is a
representation of $\mathfrak{S}_4$ on $6$ elements.
\end{example}

Let $S$ be a dense birecurrent set of degree $d$ and let $\A=(Q,i,T)$
be its minimal automaton. By Theorem~\ref{theoremRev}, the set $T$
is saturated by a word $w$ of minimal rank. The rank
of $w$ is by definition the degree $d$ of $S$. Let $k\ge 1$
be the number of classes of the kernel of $w$ whose union
is $T$.
We define the \emph{index} of a dense birecurrent set, denoted $i(S)$, as
the rational number $d/k$.

Since $1\le k\le d$, we have $1\le i(S)\le d$.

Note that the index does not depend on the choice of the word $w$.
Indeed, since $S$ is dense, any minimal image (that
is the image of a word
of minimal rank) is a set of representatives of any kernel
of a word of minimal rank and thus $k$ is
the number of elements of $T$ in each minimal image.

When $S=X^*$ with $X$ a maximal bifix code, the index of $S$
 is the degree of $X$.

\begin{example}
Let $S$ be the dense birecurrent set of Example~\ref{exampleLeftRightDegrees}.
Since the degree of $S$ is $4$, since $\varphi_\A(A^*)$ is
a group and since $\Card(T)=2$, we have $i(S)=2$.
\end{example}
\begin{example}
Let $S$ be the birecurrent set of Example~\ref{example3/2}. Then
$i(S)=3/2$.
\end{example}

\begin{proposition}\label{propositionIndex}
Let $S$ be a dense birecurrent set and let $\A$ be the minimal
automaton of $S$. Then for every $\GH$-class $H$ of the minimal
ideal of $\varphi_\A(A^*)$, one has
\begin{equation}
i(S)=\Card(H)/\Card(H\cap\varphi_\A(S)).\label{eqIndex}
\end{equation}
\end{proposition}
\begin{proof}
Set $\A=(Q,i,T)$ and let $d$ be the degree of $S$.
Let $x$ be a word of rank $d$ which saturates
$T$. Let $H$ be the $\GH$-class of $x$
and let $k$ be the number of classes of the kernel of $x$ whose
union is $T$. Let $I$ be the common image of the elements
of $H$. Since $I$ 
is a system of representatives of the kernel of $x$, it
 contains $k$ elements of $T$. Let $j$ be the element of $I$ such
that $ih=jh$ for every $h\in H$. Then, for every $h\in H$,
one has  $ih\in T$ if and only if $jh\in T$. Thus the set $H\cap\varphi_\A(S)$
is a union of $k$ cosets of the subgroup of $H$ fixing $j$.
Thus $\Card(H\cap\varphi_\A(S))=k\Card(H)/d=\Card(H)/i(S)$
\end{proof}
Note that, as a consequence, the index of a dense birecurrent set
and of its reversal are the same (contrary to the degree). Indeed,
the monoids $\varphi_\A(A^*)$ and $\varphi_{\tilde{A}}(A^*)$
are antiisomorphic and  antiisomorphic  groups
are isomorphic.

By Equation~\eqref{eqIndex}, the index of a birecurrent set
is a measure of its size.
We will make this idea more precise using probabilities.
We begin with some definitions on Bernoulli distributions
(see~\cite{BerstelPerrinReutenauer2009} for more details).

Let $\pi$ be a positive Bernoulli distribution on $A^*$. By definition, $\pi$
is a morphism from $A^*$ into the interval $]0,1]$ such that
$\sum_{a\in A}\pi(a)=1$. For $S\subset A^*$, we denote by $\pi(S)$
the (possibly infinite) sum $\sum_{x\in S}\pi(x)$.

For any code, one has $\pi(X)\le 1$
and when $X$ is a thin maximal code,
one has $\pi(X)=1$ (see \cite[Theorem 2.5.19]{BerstelPerrinReutenauer2009}). Moreover, if $X$ is prefix, then
we have $\lambda(X)=\pi(P)$ where $\lambda(X)=\sum_{x\in X}|x|\pi(x)$
is the \emph{average length} of $X$
and where $P$ is the set of proper prefixes of $X$ (see~\cite[Corollary 3.7.13]{BerstelPerrinReutenauer2009}).

The \emph{density}
of a recognizable set $S$, denoted $\delta(S)$
 is the Cesaro limit of the numbers
$\pi(S\cap A^n$). Thus $\delta(S)=\lim\frac{1}{n}\sum_{0\le i<n}\pi(S\cap A^i)$.
By  \cite[Theorem 13.2.9]{BerstelPerrinReutenauer2009},
when $X$ is a thin maximal code, we have $\delta(X^*)=1/\lambda(X)$.

The following result shows that the density of a dense birecurrent set
is a rational number, independent of the Bernoulli ditribution $\pi$.
This surprising property was  first put
in evidence for recognizable bifix codes
in~\cite{Schutzenberger1961}.
\begin{theorem}\label{theoremDensity}
The density of a dense birecurrent set is the inverse of its index.
\end{theorem}
\begin{proof}
Let $S$ be a dense birecurrent set and let $\A$ be its
minimal automaton. Set $\varphi=\varphi_A$ and $M=\varphi(A^*)$,
and let $K$ be the minimal ideal of $M$.
Let $\nu=\delta\varphi^{-1}$ where $\delta$
denotes the density. Then, by
 \cite[Theorem 13.4.7]{BerstelPerrinReutenauer2009},
 $\nu$ is a probability
measure on the family of subsets of $M$. Moreover, $\nu(K)=1$
and for every $m\in K$, one has
\begin{equation}
\nu(m)=\frac{\nu(H)}{\Card(H)}
\end{equation}
where $H$ is the $\GH$-class of $m$. Let $\cal F$ denote
the family of $\GH$-classes of $K$. Then, using Proposition~\ref{propositionIndex},
\begin{eqnarray*}
\delta(S)&=&\sum_{m\in \varphi(S)\cap K}\nu(m)\\
&=&\sum_{H\in\cal F}\Card(H\cap \varphi(S))\frac{\nu(H)}{\Card(H)}\\
&=&\frac{1}{i(S)}\sum_{H\in\cal F}\nu(H)=\frac{\nu(K)}{i(S)}=\frac{1}{i(S)}.
\end{eqnarray*}
\end{proof}
Theorem \ref{theoremDensity} implies the following result
for a dense birecurrent set of finite type.
\begin{corollary}\label{corollaryLambda}
Let $S=X^*P$ be a dense birecurrent set of finite type
with left root $X$. Then, for any positive Bernoulli distribution $\pi$,
\begin{displaymath}
\lambda(X)=i(S)\pi(P).
\end{displaymath}
\end{corollary}
\begin{proof}
Since $S=X^*P$, we have $\delta(S)=\delta(X^*)\pi(P)$ by \cite[Proposition 13.2.5]{BerstelPerrinReutenauer2009}. Since $X$ is a finite maximal
prefix code, we have $\delta(X^*)=1/\lambda(X)$ by \cite[Theorem 13.2.9]{BerstelPerrinReutenauer2009}. Thus, by Theorem \ref{theoremDensity} we have
$\lambda(X)=\pi(P)/\delta(S)=i(S)\pi(P)$.
\end{proof}
Note that when Equation~\eqref{eqConjecture} holds, we have
$\pi(M)=i(S)$. The fact that $\pi(M)$ is a rational number
independant of $\pi$ is itself a consequence of the
equation $\u(M)(1-\u(A))=(1-\u(A))\u(N)$ which implies
by left Euclidean division
$\u(M)=\alpha+(1-\u(A))u$ for polynomial $u$ and some scalar $\alpha$
and thus $\pi(M)=\alpha$.

By a well-known result, for every $d\ge 1$
and any finite alphabet,
 there is only a finite
number of finite maximal bifix codes of degree 
$d$ on this alphabet (see~\cite[Theorem 6.5.2]{BerstelPerrinReutenauer2009}).
There is no similar property for dense birecurrent sets of finite
type, as shown by the following example.
\begin{example}
For $n\ge 2$, let $Z=\{a^n,a^{n-1}b,\ldots,ab,b\}$, let $X=Z^2$ and let $P=\{\varepsilon,a,\ldots,a^{n-1}\}$. Then $S=X^*P$ is a birecurent set 
(the case $n=2$ is Example~\ref{examplePalindrome}). Indeed, we have
$\u(Z)\u(P)=\u(P)\u(\tilde{Z})$ and thus $\tilde{S}=\u(P)\u(\tilde{X})^*$.
The degree of $S$ is $2$ for every $n\ge 1$ because 
$(ZP\cup P)b\subset X\cup Z$ and thus the rank of $b$ is $2$.
\end{example}

\subsection{Indecomposable prefix codes}\label{sectionIndecomposable}
In this section, we relate birecurrent sets with the notion
of decomposition of prefix codes.

A prefix code $X\subset A^*$ is \emph{indecomposable} if $X\subset Z^*$, with
$Z\subset A^*$ a prefix code, implies $Z=A$ or $Z=X$. Otherwise
$X$ is said to be \emph{decomposable} over $Z$(see~\cite{BerstelPerrinReutenauer2009} for a more detailed
presentation).

If $X$ is decomposable over $Z$, let $\beta:B\rightarrow Z$ be a bijection
of $Z$ with an alphabet $B$, extended to a morphism from $B^*$
onto $Z^*$. Then $Y=\beta^{-1}(X)$ is a prefix code on the alphabet $B$.
We denote $X=Y\circ_\beta Z$. The prefix code $X$ is  thin maximal, if
and only if $Y$
and $Z$ are also thin maximal prefix codes. Moreover, one has $d(X)=d(Y)d(Z)$
\cite[Proposition 11.1.2]{BerstelPerrinReutenauer2009}.
In particular, $X$ is synchronized (that is, of degree $1$)
if and only if $Y$ and $Z$ are synchronised.

We will prove the following result. It singles out two basic building
blocks for recognizable maximal prefix codes: synchronized ones
on the one hand, and left roots of dense birecurrent sets on the other.
Note that no nontrivial prefix code (that is, distinct of the alphabet)
can belong to both families.
\begin{theorem}\label{theoremDecomposition}
Let $X$ be a recognizable maximal prefix code. If $X$ is indecomposable,
either $X$ is synchronized or it is the left root of a dense
birecurrent set.
\end{theorem}
To prove Theorem~\ref{theoremDecomposition}, we introduce the following notion,
which plays a role in the solution of the road coloring problem (see
\cite[Lemma 10.4.3]{BerstelPerrinReutenauer2009}).

Let $\A=(Q,i,T)$ be a finite deterministic 
automaton. A pair of states $p,q\in Q$
is called \emph{synchronizable} if there is a word $v$ such that
$p\cdot v=q\cdot v$. It is called \emph{strongly synchronizable}
if for every $u\in A^*$ the pair $p\cdot u,q\cdot u$ is synchronizable.

We note that the equivalence $\rho$ on $Q$ defined by $p\equiv q \bmod \rho$ if $p,q$ are strongly synchronizable is a \emph{stable equivalence}.
This means that if $p\equiv q \bmod \rho$, then $p\cdot u\equiv q\cdot u\bmod \rho$ for every word $u$.

\begin{proposition}
Let $\A$ be a strongly connected and complete finite deterministic automaton.
Two states $p,q$ of $\A$
are strongly synchronizable if and only if $p\cdot x=q\cdot x$
for every word $x$ of minimal rank.
\end{proposition}
\begin{proof}
The condition is necessary. Indeed, let $p,q$ be strongly synchronizable
and let $x$ be a word of minimal rank. Let $y$ be a word such that
$p\cdot xy=q\cdot xy$. Since $\varphi_\A(x)$ generates a minimal
right ideal, there is a word $z$ such that $\varphi_\A(xyz)=\varphi_\A(x)$.
Thus $p\cdot x=p\cdot xyz=q\cdot xyz=q\cdot x$.

Conversely, assume that the condition is satisfied. Let $x$ be a word of minimal rank. For every word $u$,
the word $ux$ has minimal rank and thus $p\cdot ux=q\cdot ux$. Thus
$p,q$ are strongly synchronizable.
\end{proof}

Let $X$ be a recognizable maximal prefix code and let $\A=(Q,i,i)$
be the minimal automaton of $X^*$. If $\rho$ is a stable
equivalence on $Q$, then the set $U=\{u\in A^*\mid i\cdot u \equiv i\bmod\rho\}$
is a submonoid generated by a prefix code $Z$
with $X\subset Z^*$. If $X=Z$, then $\rho$ must be the equality
since $\A$ is minimal. Thus, if $X$
is indecomposable, $\rho$ must be the equality.\\

\begin{proofof}{of Theorem~\ref{theoremDecomposition}}
Assume that $X$ is not synchronized. Let $\A=(Q,i,i)$ be the minimal
automaton of $X^*$. Let $x$ be a word of minimal rank and let
$T$ be a class of the kernel of $x$. Let $\A'=(Q,i,T)$.
The set $S$ recognized by $\A'$ is birecurrent by Theorem~\ref{theoremRev}.
Let us verify that
$\A'$  is minimal. This will imply our conclusion since then
$X$ is the left root of the birecurrent set $S$.

We first observe that since $X$ is indecomposable, two strongly
synchronizable states are equal. Indeed, this follows from the
observation made above that the equivalence on $Q$ defined
by $p\equiv q$ if $p,q$ are strongly synchronizable is a stable equivalence.

Let $p,q$ be two states such that for every word $w$, $p\cdot w\in T$
if and only if $q\cdot w\in T$. Let $u$ be a word. Since $\A$
is strongly connected, there is some word $v$ such that $p\cdot uv\in T$.
Then $q\cdot uv\in T$ and thus $p\cdot uvx=q\cdot uvx$. This
shows that $p,q$ are strongly synchronizable. Since $X$
is indecomposable, this implies $p=q$. Thus $\A'$ is minimal.
\end{proofof}
Theorem~\ref{theoremDecomposition} is related with an old conjecture
of Sch\"utzenberger asserting that if a finite maximal prefix
code is indecomposable, either it is synchronized or it is bifix.
The conjecture is not true as shown by the left
root $X$ of the birecurrent set $S$ of Example~\ref{exampleDegree3}
(which appeared originally in~\cite{Perrin1977}).
In fact $X$ is indecomposable (see below), is not bifix and not synchronized
since its degree is $3$.
\begin{example}
Let us show that the left root $X$ of the birecurrent set of Example~\ref{exampleDegree3} is indecomposable and, more generally, that if $Z$ is a finite
maximal bifix code of prime degree $d\ge 3$ and $w$ is a pure square
for $Z$, the prefix code
$X$ given by Equation~\eqref{eqDelta} is indecomposable
(this is already proved in~\cite{Perrin1977}, but we reproduce 
the proof for convenience). 

We first observe that $a^d\in X$ for all $a\in A$. Indeed, one has $a^d\in Z$
as for any finite maximal bifix code of degree $d$ \cite[Proposition 6.5.1]{BerstelPerrinReutenauer2009}. Next, since $d\ge 3$, it is odd and
thus we cannot have $w\in a^*$.

Let
$T$ be a prefix code such that $X\subset T^*$. Then, since $d$ is prime,
 one has for any
$a\in A$, either $a\in T$ or $a^d\in T$. Set $X=Y\circ_\beta T$. Since
$d(X)=d(Y)d(T)$, either $Y$ is synchronized or $T$ is synchronized.
We show that
$T=A$ or $T=X$.

Assume first that $T$ is
synchronized. Then $Y$ has degree $d$. If $a^d\in T$ for some $a\in A$,
then $Y$ is synchronized (some power of $b=\beta^{-1}(a^d)$ is a synchronizing
word). This is impossible since $d>1$.
Therefore $a\in T$ for every $a\in A$ and thus $T=A$. 

Assume next that $Y$ is synchronized.
Then, since $T$ has degree $d$, $a^d\in T$ for every $a\in A$. 
Fix some $a\in A$.
By inspection of Equation~\ref{eqDelta}, the suffixes of $X$
which are in $a^*$ are of length at most $d$ and the only proper
prefixes $p$ of $X$ such that $pa^d\in X$ are $\varepsilon$ and $w$.
This implies that the only proper prefixes
of $X$ which possibly belong to $T^*$ are $\varepsilon$ and $w$.
Indeed, if $p\in T^*$ and $a^n\in X$, then $a^n\in T^*$
and thus $n$ is multiple of $d$.
 But if $w\in T^*$,
we have $D\subset T^*$ since $w^3D\subset X$ which is impossible since
the integer $n$ such that $a^n\in D$ is strictly less than $d$.
 This forces $T=X$ and shows that $X$
is indecomposable.

Note that the statement is not true for $d=2$ since the code
$X$ of Example \ref{examplePalindrome3} is decomposable.

\end{example}

We formulate the following open problem, as
an attempt to replace the conjecture of Sch\"utzenberger
by a weaker statement: if the prefix code
in Theorem~\ref{theoremDecomposition} is additionnaly finite,
can one prove that either it is synchronized or it is the left root of a dense
birecurrent set of finite type?

\section{Complete reducibility}\label{sectionCompleteReducibility}
In this section, we develop the link between birecurrent sets
and completely reducible sets. We begin with an introduction to
formal series (see~\cite{BerstelReutenauer2011} 
or~\cite{Sakarovitch2009} for a more detailed presentation).
\subsection{Recognizable series}\label{sectionFormalSeries}
We consider  formal series with coefficients in the field $K=\Q$.
Recall from Section~\ref{sectionPreliminaries} that we denote
by $K\<<A>>$ the ring of formal series with coefficients in $K$
and noncommutative variables in $A$.

A series is \emph{recognizable} (or equivalently, by
Sch\"utzenberger's theorem, \emph{rational})
if there is a linear representation
over a finite dimensional space recognizing it.
There is  a unique linear representation of
minimal dimension recognizing a given recognizable series
(up to the choice of a basis if the representation
is given in matrix terms). One can compute
it following three steps.
\begin{enumerate}
\item Start from any linear representation $(\lambda,\mu,\gamma)$ recognizing
 $\sigma$.
\item Take the representation $(\lambda',\mu',\gamma')$ obtained
by taking $\lambda'=\lambda$ and by restricting $\mu$ and $\gamma$ to the subspace generated by the vectors $\lambda\mu(w)$
for $w\in A^*$.
\item Take the representation $(\lambda'',\mu'',\gamma'')$ obtained
by taking $\gamma''=\gamma'$ and by restricting $\mu'$ and $\lambda'$
to the subspace generated by the vectors $\mu'(w)\gamma'$
for $w\in A^*$.
\end{enumerate}

\begin{example}\label{exampleaplus}
Let $A=\{a\}$.
The linear representation 
\begin{displaymath}
\lambda=[1\ 0],\quad \mu(a)=\begin{bmatrix}0&1\\0&1\end{bmatrix},\quad
\gamma=\begin{bmatrix}0\\1\end{bmatrix}
\end{displaymath}
recognizes the characteristic series of the set $a^+$. It is minimal.
Choosing $\gamma=[1\ 1]^t$, the representation recognizes $a^*$.
It is not minimal because $\mu(a)\gamma=\gamma$. It is easy to see that
the minimal
representation of $a^*$ is of dimension $1$.
\end{example}

The computation of the minimal representation of the
characteristic series of a set $S\subset A^*$ is closely
related to the computation of the minimal automaton
of $S$ and of its reversal. It can be described as follows.

Consider the
vector space $K^Q$ as containing $Q$, identifying $q\in Q$
with its characteristic function. Let
$(\lambda,\mu,\gamma)$ be the following linear representation
on the space $V=K^Q$. Set $\lambda=i$. For $w\in A^*$,
define $\mu(w)$ as the endomorphism of $V$ such that
$q\mu(w)=q\cdot w$. Finally set $\gamma=\u(T)$ where
$\u(T)$ is the linear form on $V$ which is the characteristic function of $T$.
Then $(\lambda,\mu,\gamma)$ obviously recognizes $\u(S)$
and we have performed the first step of the algorithm given above
to compute the minimal representation.
Step 2  does not change the representation since each state of the
minimal automaton is accessible from the initial state.
To perform Step 3, we compute the set $\tilde{Q}$ of states
of the reversal $\tilde{\A}^\delta$ of $\A$. For $w\in A^*$,
the  vector $\mu(w)\u(T)$ is precisely $\u(w\cdot T)$. Thus,
the minimal representation is the restriction of $\mu$ to
the subspace of $K^Q$ generated by the vectors $\u(U)$ for
$U\in\tilde{Q}$ (considered as a column vector on which
each $\mu(w)$ acts on the left).

We illustrate this algorithm in the following examples. 
In the first one, the representation given by the minimal
automaton is the minimal one.

\begin{example}\label{exampleaplus1}
The minimal automaton $\A$ of $S=a^+$ is represented in Figure~\ref{figureAutomatonaplus}.
\begin{figure}[hbt]
\centering
\gasset{Nadjust=wh}
\begin{picture}(20,10)
\node[Nmarks=i](1)(0,0){$1$}\node[Nmarks=f](2)(10,0){$2$}

\drawedge(1,2){$a$}\drawloop(2){$a$}
\end{picture}
\caption{The minimal automaton of $a^+$ in Example~\ref{exampleaplus1}.}\label{figureAutomatonaplus}
\end{figure}
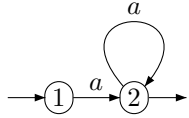
The linear representation built from the automaton $\A$ is that
of Example~\ref{exampleaplus}. The states
of $\tilde{\A}^\delta$ are $\{2\}$ and $\{1,2\}$. Their characteristic functions
are linearly independent. Thus Step 3 does not modify
the representation, which is minimal.
\end{example}
In the second example, the representation given by the minimal
automaton is not minimal.
\begin{example}\label{exampleBifix3}
Let $X$ be the bifix code represented in Figure~\ref{fig3_03} on the right
(it is the unique maximal bifix code of degree $3$ with internal part
$ab$). The minimal automaton $\A=(Q,1,1)$ of $S=X^*$ is defined by its
transitions given in Table~\ref{tableBifix3}.
\begin{table}[hbt]
\begin{displaymath}
\begin{array}{c|ccccc}
 & 1& 2& 3& 4& 5\\ \hline
a& 2& 3& 1& 3& 1\\
b& 4& 1& 5& 5& 1
\end{array}
\end{displaymath}
\caption{The transitions of the minimal automaton $\A$
in Example~\ref{exampleBifix3}.}\label{tableBifix3}
\end{table}

The set $\tilde{Q}$ of states of the reversal automaton 
$\tilde{\A}^\delta$ are given in Table~\ref{tableRevBifix3}.
\begin{table}[hbt]
\begin{displaymath}
\begin{array}{c|ccccc}
\tilde{Q}&\mathbf{1}&\mathbf{2}&\mathbf{3}&\mathbf{4}&\mathbf{5}\\\hline
         &1         & 2,4      &  3,5     & 2,5      &  3,4
\end{array}
\end{displaymath}
\caption{The states of the automaton $\tilde{\A}^\delta$
in Example~\ref{exampleBifix3}.}\label{tableRevBifix3}
\end{table}

The vector space generated by the corresponding characteristic vectors
has dimension $4$ because $\u(\mathbf{2})+\u(\mathbf{3})=\u(\mathbf{4})+\u(\mathbf{5})$. 
Choosing the basis formed of $\u(\mathbf{1}),\u(\mathbf{2}),\u(\mathbf{3}),
\u(\mathbf{4})$ the minimal representation $(\lambda,\mu,\gamma)$ of $\sigma=\u(S)$
is 
\begin{displaymath}
\lambda=[1\ 0\ 0\ 0],\quad 
\mu(a)=\begin{bmatrix}0&1&0&1\\0&0&1&0\\1&0&0&0\\0&0&0&0\end{bmatrix},\quad
\mu(b)=\begin{bmatrix}0&1&0&0\\0&0&1&1\\0&0&1&1\\1&0&-1&-1\end{bmatrix},\quad
\gamma=\begin{bmatrix}1\\0\\0\\0\end{bmatrix}
\end{displaymath}
Let us verify for example the value of the first column of $\mu(a)$.
We have $\mu(a)\u(\mathbf{1})=\u(a\cdot\{1\})=\u(\{3,5\})=\u(\mathbf{3})$.
Similarly, the last column of $\mu(b)$ is computed as
$\mu(b)\u(\mathbf{4})=\u(b\cdot\{2,5\})=\u(\{3,4\})=\u(\mathbf{5})=
\u(\mathbf{2})+\u(\mathbf{3})-\u(\mathbf{4})$.
\end{example}

Let $\mu:A^*\rightarrow\End(V)$ be a morphism. We let the
endomorphisms of $V$ act on the left
of the vectors in $V$. A subspace $W$
of $V$ is said to be \emph{invariant} with respect to $\mu$
if for every $x\in W$ and
$w\in A^*$, one has $x\mu(w)\in W$. The vector space $V$ is said
to be \emph{irreducible} with respect to $\mu$
if $V\ne 0$ and if its only invariant subspaces are
$0$ and $V$. Finally, it is said to be \emph{completely reducible}
with respect to $\mu$
if $V=\oplus_{i=1}^nV_i$ where each $V_i$ is an invariant irreducible
subspace of $V$. Let $\mu_i(w)$ be the restriction of $\mu(w)$
to $V_i$. The representation $\mu$ is  the direct
sum of the representations $\mu_i$ which are called the
\emph{irreducible components} of $\mu$.

We also say that a linear representation $(\lambda,\mu,\gamma)$
is completely reducible
if the underlying vector space is completely reducible with respect
to $\mu$.
\begin{example} \label{exampleaplus2}
The linear representation of Example~\ref{exampleaplus} is completely reducible.
Indeed, the subspaces generated respectively by $[-1\ 1]$
and by $[0\ 1]$ are invariant and obviously irreducible. In the
basis formed by these vectors, the representation takes the 
following equivalent form.
\begin{displaymath}
\lambda'=[-1\ 1],\quad \mu'(a)=\begin{bmatrix}0&0\\0&1\end{bmatrix},
\quad \gamma'=\begin{bmatrix}1\\1\end{bmatrix}.
\end{displaymath}
\end{example}
For $u\in A^*$, we denote
by $\sigma\cdot u$ the series defined by $(\sigma\cdot u,v)=(\sigma,uv)$ for every $v\in A^*$.

The \emph{syntactic space} of a series $\sigma$, denoted $V_\sigma$
is the vector space generated by the series $\sigma\cdot u$ for all $u\in A^*$.
The \emph{syntactic representation} of $\sigma$ is the 
linear representation
$\psi_\sigma:K\langle A\rangle\rightarrow \End(V_\sigma)$ defined
for $\tau\in V_\sigma$ and $u\in A^*$ by
\begin{displaymath}
\tau\psi_\sigma(u)=\tau\cdot u.
\end{displaymath}
The linear representation defined by the triple $(\lambda,\psi_\sigma,\gamma)$
with $\lambda=\sigma$ and $\gamma$ defined by $\tau\gamma=(\tau, \varepsilon)$
recognizes $\sigma$. Indeed, one has for every $w\in A^*$,
\begin{displaymath}
\lambda\psi_\sigma(w)\gamma=(\sigma\cdot w,\varepsilon)=(\sigma,w).
\end{displaymath}
It can be shown that the syntactic representation of a recognizable series
is its minimal representation (see~\cite{BerstelReutenauer2011}).

If a series has a completely reducible representation,
then its  syntactic representation  is also completely reducible.

\begin{example}\label{exampleaplus3}
The syntactic space of the series $\sigma=\u(a^+)$ has dimension $2$ and a basis
is formed by $\sigma$ and $\sigma\cdot a=\u(a^*)$. The corresponding linear representation is given in Example~\ref{exampleaplus}.
\end{example}

\subsection{Completely reducible sets}\label{sectionComplelyReducible}

A set $S\subset A^*$ is \emph{completely reducible} over $K$ if the
syntactic representation of the series $\underline{S}$ is
completely reducible. Equivalently, its syntactic algebra 
is semi-simple \cite[Chapter 12]{BerstelReutenauer2011}.

\begin{example}
The set $a^*$ is completely reducible since its syntactic space has dimension $1$.
The set $S=a^+$ is also completely reducible since its syntactic representation
is completely reducible by Examples \ref{exampleaplus2} and \ref{exampleaplus3}.

\end{example}
We give a second example, in which $S$ is the submonoid generated by a finite
bifix code.
\begin{example}
Consider again the set $S$ of Example~\ref{exampleBifix3}.
We have found a minimal representation  of dimension $4$.
This representation is not irreducible because the space generated by
$\u(\mathbf{1})+\u(\mathbf{2})+\u(\mathbf{3})$ is invariant by $\mu$ (with $\mu(w)$ acting on the left).
The space generated by $\u(\mathbf{2})-\u(\mathbf{1})$, $\u(\mathbf{3})-\u(\mathbf{1})$ and $\u(\mathbf{4})-\u(\mathbf{1})$
is a stable complement and the representation takes in the
basis $\u(\mathbf{1})+\u(\mathbf{2})+\u(\mathbf{3})$, $\u(\mathbf{2})-\u(\mathbf{1})$,  $\u(\mathbf{3})-\u(\mathbf{1})$, $\u(\mathbf{4})-\u(\mathbf{1})$
the form of a direct sum of two representations $\mu'_1$ of dimension $1$
and $\mu'_2$ of dimension $3$. In this basis, the vector $\lambda$
becomes $\lambda'=[1\ -1\ -1\ -1]$ (its components are the values of the
linear form defined by $\lambda$ on each vector of the basis), and $\mu(a)$, $\mu(b)$, $\gamma$
become
\begin{displaymath}
\mu'(a)=\begin{bmatrix}1&0&0&0\\0&0&1&0\\0&-1&-1&-1\\0&0&0&0\end{bmatrix},\quad
\mu'(b)=\begin{bmatrix}1&0&0&0\\0&0&1&1\\0&0&1&1\\0&-1&-2&-2\end{bmatrix},\quad
\gamma'=\begin{bmatrix}1/3\\-1/3\\-1/3\\0\end{bmatrix}
\end{displaymath}
The value of $\gamma'$ results from the computation of
$\gamma=\u(\mathbf{1})=\frac{1}{3}(\u(\mathbf{1})+\u(\mathbf{2})+\u(\mathbf{3}))-\frac{1}{3}(\u(\mathbf{2})-\u(\mathbf{1}))
-\frac{1}{3}(\u(\mathbf{3})-\u(\mathbf{1}))$. The representations $\mu'_1$ and $\mu'_2$
are irreducible. This is obvious for $\mu'_1$. For $\mu'_2$ it
can be proved directly by verifying that the matrices $\mu'_2(w)$
generate the algebra of $3\times 3$-matrices over $K$ or deduced
as a consequence of \cite[Theorem 2.2]{Perrin2013}).
\end{example}

Let $S\subset A^*$ be a recognizable set, let $\sigma=\u(S)$ and let $\A=(Q,i,T)$ be its
minimal automaton. Set $\varphi=\varphi_\A$, $\psi=\psi_\sigma$
 and $M=\psi(A^*)$.
By \cite[Proposition 14.7.1]{BerstelPerrinReutenauer2009} for all $u,v\in A^*$, one has 
\begin{equation}
\varphi(u)=\varphi(v)\Leftrightarrow \psi(u)=\psi(v).\label{eqPhiPsi}
\end{equation}
In particular, $M$ and $\varphi(A^*)$ are isomorphic.

An element of $M=\psi(A^*)$ is a linear map from $V_\sigma$ into itself
and, as such, it has a kernel and an image which are subspaces of $V_\sigma$.

The \emph{eventual kernel} of $S$, denoted $EK(S)$, is the intersection of the kernels
of all elements of minimal nonzero rank of $M$.

Symmetrically, the \emph{eventual range} of $S$, denoted
$ER(S)$, is the subspace spanned by
 the images of  all elements of minimal nonzero rank of $M$.

Both $EK(S)$ and $ER(S)$ are invariant subspaces of $V_\sigma$. Indeed,
let $x\in EK(S)$ and let $w\in A^*$. For any $m\in M$ of minimal
nonzero rank, $\psi(w)m$ is either zero or has minimal nonzero rank .
In both cases $x\psi(w)m=0$. Thus $x\in EK(S)$. Similarly, let 
$x\in ER(S)$ and let $w\in A^*$. Since $x\in ER(S)$, we have
$x=\sum_{m\in J}x_mm$ where $J$ denotes the set of elements of minimal
nonzero rank in $M$ and $x_m\in V_\sigma$. Then
$x\psi(w)=\sum_{m\in J}v_mm\psi(w)$ and thus $x\psi(w)\in ER(S)$.

\begin{theorem}\label{theoremCharactCompleteRed}
A recurrent set is completely reducible
if and only if $EK(S)=\{0\}$ .
\end{theorem}

We first prove the following statement.

\begin{proposition}\label{propInter}
Let $S$ be a recognizable set. If $S$ is completely reducible,
then $EK(S)\cap ER(S)=\{0\}$.
\end{proposition}
\begin{proof}
Set $\sigma=\u(S)$ and $V=V_\sigma$. Let $J$ be the set
of elements of $M=\psi(A^*)$ of nonzero minimal rank.
Set $W=EK(S)\cap ER(S)$.
Since $W$ is invariant, and since $S$ is completely
reducible, there is a complement $W'$ of $W$ which is an invariant subspace
of $V$.  

Let $v\in W$. Then, as for any element of $ER(S)$,
we have $v=\sum_{m\in J}\alpha_m(\sigma m)$
for some $\alpha_m\in K$. Since $V=W\oplus W'$, there exist
$\rho\in W$ and $\rho'\in W'$ such that $\sigma=\rho+\rho'$.
Then
\begin{eqnarray*}
v&=&\sum_{m\in J}\alpha_m(\sigma m)=\sum_{m\in J}\alpha_m(\rho m)+\sum_{m\in J}\alpha_m(\rho' m)\\
&=&\sum_{m\in J}\alpha_m(\rho' m).
\end{eqnarray*}
Indeed, since $\rho\in W$, we have $\rho m=0$ for every $m\in J$.
This implies that $v\in W'$, and finally $v=0$.
Therefore we conclude that $EK(S)\cap ER(S)=\{0\}$.
\end{proof}

\begin{proofof}{ of Theorem~\ref{theoremCharactCompleteRed}}
Let $\A=(Q,i,T)$
be the minimal automaton of the recurrent set
 $S$. 
Set $\varphi=\varphi_\A$, $\sigma=\u(S)$,
 $\psi=\psi_\sigma$ and $M=\psi(A^*)$.

Assume first that the eventual kernel of $S$ is $0$.

It is enough to prove
the property under the hypothesis that $S$ contains the empty word. 
Indeed,  Let $S'$ be the set
recognized by the automaton $\A'=(Q,t,T)$ for some $t\in T$. Then
$S'$ is a  set recognized by a strongly
connected deterministic automaton and $S'$ contains the empty word. Since
$\A$ is strongly connected, $S$ is a residual of $S'$ and $S'$
is a residual of $S$. Thus the syntactic representations of $S$ and $S'$
only differ by the choice of the initial vector. Thus $S$ is completely reducible
whenever $S'$ is
and the eventual kernel of $S'$ is also $0$.

 Since $\A$ is strongly connected, the monoid
$\varphi(A^*)$ has a unique $0$-minimal ideal $K$ which is a regular
$\GD$-class plus $0$ (see Section~\ref{sectionBirecurrent}).
 Let $x\in A^*$ be such that $\varphi(x)$ is an idempotent
of $K$. 
We may assume that $i\in Q\cdot x$.
Indeed, since $\A$ is strongly connected, the state $i$ is, as any state of $Q$,  in the image of some word $w$
of minimal nonzero rank. Since $K\setminus 0$ is regular $\GD$-class, the $\GL$-class of $w$ contains
an idempotent which has the same image as $w$.
Since $\varphi(x)$ is idempotent, $i\in Q\cdot x$ implies that $i\cdot x=i$.
Since $M=\psi(A^*)$ is isomorphic to $\varphi(A^*)$,
$e=\psi(x)$ is an idempotent of $M$.

Set $V=V_\sigma$. We verify that the conditions of \cite[Corollary 1]{Perrin2013} are satisfied by $V$ and $e$.
\begin{enumerate}
\item[(i)] The subspace $Ve$ is completely reducible with respect
to the restriction of $\psi$ to $xA^*x$. This is true by Maschke's Theorem
asserting that
any linear representation
of a finite group on a field with characteristic zero is completely reducible.
Indeed, $\psi(xA^*x)$ is a finite group or a finite group union zero.
\item[(ii)] The space $V$ is generated by $VeM$. Indeed
$i\cdot x=i$ implies 
$\lambda e=\lambda$. Thus $VeM$ contains $\lambda eM=\lambda M$ which generates
$V$.

\item[(iii)] We have $\{v\in V\mid vMe=0\}=0$. Indeed, $vme=0$ if and
only if $v$ is in the kernel of $me$. Since $K$ is $0$-minimal,
we have $K=MeM$ and thus  the kernel
of any element of minimal rank is equal to the kernel
of some $me$. Hence we have $vMe=0$ if and only if
$v$ belongs to the intersection of the
kernels of the elements of $K$.
\end{enumerate}
Thus we can conclude that $\psi$ is completely reducible.

Conversely, since $S$ is recurrent, we have $ER(S)=V$
and thus $EK(S)=\{0\}$ by Proposition~\ref{propInter}.
\end{proofof}
We obtain as a corollary \cite[Theorem 5.2]{Perrin2013}, which
is a generalization of the result of~\cite{Reutenauer1981}
asserting that the submonoid generated by a bifix code is completely reducible.
\begin{corollary}\label{corollaryCompleteRed}
Any recognizable birecurrent set is completely reducible.
\end{corollary}
\begin{proof}
Let $\A=(Q,i,T)$ be the minimal automaton of the set $S$.
 Set $V=K^Q$ and let $(\lambda,\mu,\gamma)$
be the linear representation associated to the automaton $\A$.
Thus $\gamma=\u(T)$. The minimal representation of $\u(S)$
is, as seen before, the restriction $\mu'$ of $\mu$ to the subspace $V'$
of column vectors generated by the vectors $\mu(w)\gamma$.
We consider $\mu'$ as a representation acting on the right
on a supplementary subspace $W$ of the ortogonal of $V'$.

By Theorem~\ref{theoremRev}, $T$ is saturated by a word $x$ of minimal
nonzero rank. Let $K$ be the $0$-minimal ideal of 
the monoid $M=\mu(A^*)$. Then $\mu(x)\in K$.
Let $e$ be an idempotent of $K$ in the $\GR$-class of $\mu(x)$
and let $y$ be such that $\mu(y)=e$. Then
$\varphi_\A(y)$ is an idempotent.
Since $\mu(x)$ and $\mu(y)$ belong
to the same $\GR$-class, $x$ and $y$ have the
same kernel. Thus $T$ is a union of classes of the kernel of $y$,
which implies $y\cdot T=T$ and thus
$e\gamma=\gamma$.

Let $v\in W$ be such that $vm=0$ for every $m\in K$.
Then, for any $m\in M$, we have $vm\gamma=vme\gamma$ since $e\gamma=\gamma$.
But since $me\in K\cup 0$, we have $vme\gamma=0$
and thus $vm\gamma=0$ for all $m\in M$. This shows that
$v$ is orthogonal to $V'$, and thus that $v=0$. 

We conclude that $EK(S)=0$ and thus that $S$ is completely
reducible by Theorem~\ref{theoremCharactCompleteRed}.
\end{proof}
Note the following precision on Theorem~\ref{theoremCharactCompleteRed}.
Let $M$ be a finite monoid having a unique $0$-minimal ideal $J$.
Then all $\GH$-classes of $J$ which are groups are isomorphic.
The \emph{Suschkevitch group} of $M$ is any of them
(see~\cite{BerstelPerrinReutenauer2009} or~\cite{RhodesSteinberg2009}).
\begin{proposition}\label{propositionNumberIrred}
Let $S$ be a recurrent completely reducible set, let $\sigma=\u(S)$ and let $\psi=\psi_\sigma$. The number of irreducible constituents of $\psi$
is equal to the number of irreducible contituents of its
restriction to the Suschkevitch group of $\psi(A^*)$.
\end{proposition}
\begin{proof}
Let $x\in A^*$ be a word such that $e=\psi(x)$
is an idempotent of the $0$-minimal ideal of the monoid
$\psi(A^*)$. Set $V=V_\sigma$. The restriction of $\psi$ to
group $G=\psi(xA^*x)$ is a representation of $G$
on $Ve$. Moreover, by \cite[Corollary 1]{Perrin2013}, if
$V=\oplus_{i=1}^mV_i$ is a decomposition of $V$ in irreducible subspaces,
then $Ve=\oplus_{i=1}^mV_ie$ is a decomposition of $Ve$ in irreducible
subspaces. Thus
the
number of irreducible components of $\psi$ is equal to the
number of irreducible components of the restriction of $\psi$ to $G$.
\end{proof}
Note that, in particular, if the Suschkevitch group of $\psi(A^*)$ is
trivial, then $\psi$ is irreducible.

The following example shows that the hypothesis of Theorem~\ref{theoremCharactCompleteRed} can be satisfied by a set which is not birecurrent.
\begin{example}\label{example34}
 Consider again  the strongly connected automaton $\A$ of
Figure~\ref{figureRev} on the left with its deterministic reversal on the right
(Example~\ref{exampleBirecurrent1}). We change the automaton
$\A$ into an automaton $\A'$ by choosing this time $T'=\{2\}$
as set of terminal states. The automata $\A'$ and $\tilde{\A'}^\delta$
are represented in Figure~\ref{figureRevBis}. Let $S'$ be the set recognized by
$\A'$. Since $\tilde{\A'}^\delta$ is not strongly connected, $S'$ is not
birecurrent. 
\begin{figure}[hbt]
\centering\gasset{Nadjust=wh}
\begin{picture}(80,20)
\put(0,0){
\begin{picture}(30,20)(0,-5)

\node[Nmarks=i](1)(0,0){$1$}\node[Nmarks=f,fangle=0](2)(20,0){$2$}

\drawloop(1){$a$}\drawedge[curvedepth=5](1,2){$b$}\drawedge[curvedepth=5](2,1){$a$}
\end{picture}
}
\put(40,0){
\begin{picture}(30,20)(0,-5)
\node[Nmarks=i,iangle=180](2)(0,0){$2$}
\node[Nmarks=f,fangle=90](1)(20,0){$1$}\node[Nmarks=f,fangle=90](12)(40,0){$1,2$}
\node[ExtNL=y,NLangle=-90,NLdist=2](2)(0,0){$\mathbf{1}$}
\node[ExtNL=y,NLangle=-90,NLdist=2,Nframe=n](1)(20,0){$\mathbf{2}$}

\drawedge(2,1){$b$}
\drawloop[loopangle=0](12){$a$}
\drawedge[curvedepth=5,eyo=2](1,12){$a$}
\drawedge[curvedepth=5,sxo=2](12,1){$b$}
\node[ExtNL=y,NLangle=-90,NLdist=2,Nframe=n](12)(40,0){$\mathbf{3}$}
\end{picture}
}
\end{picture}
\caption{The automata $\A'$ and $\tilde{\A'}^\delta$ in Example~\ref{example34}.}\label{figureRevBis}
\end{figure}
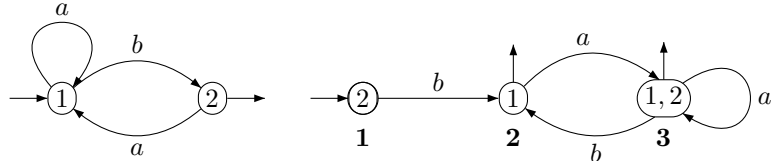
\begin{figure}[hbt]
\begin{displaymath}
\def\rb{\hspace{2pt}\raisebox{0.8ex}{*}}\def\vh{\vphantom{\biggl(}}
    \begin{array}%
    {r|@{}l@{}c|@{}l@{}c|}%
    \multicolumn{1}{r}{}&\multicolumn{2}{c}{1}&\multicolumn{2}{c}{2}\\
    \cline{2-5}
    1,2&\vh\rb &a &\vh\rb  &ab\\
    \cline{2-5}
   1&\vh\rb &ba & & b \\
   \cline{2-5}

    \end{array}
\qquad
\def\rb{\hspace{2pt}\raisebox{0.8ex}{*}}\def\vh{\vphantom{\biggl(}}
    \begin{array}%
    {r|@{}l@{}c|@{}l@{}c|}%
    \multicolumn{1}{r}{}&\multicolumn{2}{c}{\mathbf{2}}&\multicolumn{2}{c}{\mathbf{1}}\\
    \cline{2-5}
    \mathbf{1},\mathbf{2}&\vh\rb &a &\vh\rb  &ab\\
    \cline{2-5}
   \mathbf{2}&\vh\rb &ba & & b \\
   \cline{2-5}

    \end{array}
\end{displaymath}
\caption{The $0$-minimal ideals of $\A$ and $\tilde{\A}^\delta$
in Example~\ref{exampleBirecurrent1}.}\label{figureMinimalIdeal01}
\end{figure}

The $0$-minimal ideal of the monoid $\varphi_\A(A^*)=\varphi_{\A'}(A^*)$ is represented
in Figure~\ref{figureMinimalIdeal01} on the left
and the set $Q'$ of states
of $\tilde{\A'}^\delta$ in Table~\ref{tableStatesa'1}.
\begin{table}[hbt]
\begin{displaymath}
\begin{array}{c|ccc}
\tilde{Q'}&\mathbf{1}&\mathbf{2}&\mathbf{3}\\\hline
         &2    &1     &1,2
\end{array}
\end{displaymath}
\caption{The states of the automaton $\tilde{\A'}^\delta$ in Example~\ref{example34}.}\label{tableStatesa'1}
\end{table}
Except $\{2\}$, all states of $\tilde{\A'}^\delta$ are classes of the kernel
of an element of minimal rank (see Figure~\ref{figureMinimalIdeal01}).
The space generated by the characteristic functions of these states
has dimension $2$, as for the same space corresponding to $\A'$.
Indeed $T=\{1\}$ is an element of $\tilde{Q'}$ and
$\u(T')$ is in the space generated by the characteristic
functions of the states of $\A$ since
$\u(\{2\})=\u(\{1,2\})-\u(\{1\})$.
Thus $S'$ is completely reducible althought it is not birecurrent.
\end{example}
We now prove a second corollary of Theorem~\ref{theoremCharactCompleteRed}
(actually of Proposition~\ref{propInter}).
It shows that a completely reducible set which is dense satisfies
a property which is well-known for the submonoid generated
by a maximal bifix code. 
\begin{corollary}\label{corollary2}
Let $S$ be a  set recognized by a strongly connected
unambiguous finite automaton
$\A$. If $S$ is
completely reducible and dense, then $\varphi_\A(S)$ meets every
$\GH$-class of the minimal ideal of the monoid $\varphi_\A(A^*)$.
\end{corollary}
\begin{proof}
We use the same notation as in the proof of Theorem~\ref{theoremCharactCompleteRed}. Since $\A$ is strongly connected and $S$ is dense,
the monoid
 $\varphi(A^*)$ does not contain $0$
and has a unique minimal ideal  which is a union of groups.
Moreover $\varphi(S)$ meets the minimal ideal of $\varphi(A^*)$.
Since $M=\psi(A^*)$ is the image of $\varphi(A^*)$ by a morphism,
it has also a unique minimal ideal which meets $\psi(S)$.
Let $J$ be the minimal ideal of the monoid $M=\psi(A^*)$. 

For a subset $H$ of $M$, denote
$\u(H)=\sum_{m\in H}m$.

Note that for any $\GH$-class $H$ of $J$, $(\sigma\u(H),\varepsilon)=\Card(H\cap \psi(S))$. Indeed, let $m\in H$ and let $x\in A^*$ be such that
$\psi(x)=m$. Then $(\sigma m,\varepsilon)=(\sigma \psi(x),\varepsilon)=(\sigma,x)$
which is $1$ if $x\in S$ and $0$ otherwise.

Let $H,K$ be two $\GH$-classes of  $J$ of the
monoid $M$ contained in the same $\GR$-class $R$. Then the vector
$\sigma \u(H)-\sigma \u(K)$
belongs to the eventual range $ER(S)$. 

Moreover, it belongs to the eventual kernel $EK(S)$. Indeed, for every $m\in J$,
we have $Hm=Km$ since both are equal to the $\GH$-class $R\cap Mm$.
Thus, by Theorem~\ref{theoremCharactCompleteRed}, we have $\sigma \u(H)=\sigma \u(K)$.

By Proposition~\ref{propInter}, this forces $\sigma \u(H)-\sigma \u(K)=0$
for all $\GH$-classes $H,K$ included in the same $\GR$-class.

Since $\tilde{S}$ satisfies the same hypotheses as $S$ (using the
automaton $\tilde{\A}$), the conclusion follows.
\end{proof}

We give below an example illustrating Corollary~\ref{corollary2}
\begin{example}\label{example35}
Consider again the deterministic automaton $\A=(Q,i,T)$ of Example~\ref{example14}. We change
the automaton $\A$ into an automaton $\A'=(Q,i,T')$ by choosing
this time $T'=\{1,2,6\}$ as set of terminal states instead of $T=\{1,2,5\}$. Let $S'$
be the set recognized by $\A'$. The minimal ideal of the monoid
$\varphi_\A(A^*)=\varphi_{\A'}(A^*)$ is represented in Figure~\ref{figureMinIdeal}.
\begin{figure}[hbt]
\begin{displaymath}
\def\rb{\hspace{2pt}\raisebox{0.8ex}{*}}\def\vh{\vphantom{\biggl(}}
    \begin{array}%
    {r|@{}l@{}c|@{}l@{}c|@{}l@{}c|}%
    \multicolumn{1}{r}{}&\multicolumn{2}{c}{1/3}&\multicolumn{2}{c}{2/4}&\multicolumn{2}{c}{5/6}\\
    \cline{2-7}
    1,4,5/2,3,6&\vh\rb &b^2 &\vh\rb  &b^2a&\vh\rb &b^2ab\\
    \cline{2-7}
   1,2,5/3,4,6&\vh\rb &ab^2 &\vh\rb & &\vh\rb\ &\\
   \cline{2-7}
   1,4,6/2,3,5&\vh\rb& a^2b^2&\vh\rb& &\vh\rb &\\
\cline{2-7}
    \end{array}
\end{displaymath}
\caption{The minimal ideal in Example~\ref{example35}.}\label{figureMinIdeal}
\end{figure}
The set $\tilde{Q'}$
of states of the automaton $\tilde{\A'}^\delta$ is represented in
Table~\ref{tableStatesa'}.
\begin{table}[hbt]
\begin{displaymath}
\begin{array}{c|ccccccc}
\tilde{Q'}&\mathbf{1}&\mathbf{2}&\mathbf{3}&\mathbf{4}&\mathbf{5}&\mathbf{6}&\mathbf{7}\\\hline
         &1,2,6     &1,4,6     &3,4,6     &2,3,5     &1,4,5     &2,3,6     &1,2,5
\end{array}
\end{displaymath}
\caption{The states of the automaton $\tilde{\A'}^\delta$.}\label{tableStatesa'}
\end{table}
Except $\{1,2,6\}$, all states of $\tilde{\A'}^\delta$ are classes of the kernel
of an element of minimal rank (see Figure~\ref{figureMinIdeal}).
The space generated by the characteristic functions of these states
has dimension $4$, as for the same space corresponding to $\A'$.
Indeed $T=\{1,2,5\}$ is an element of $\tilde{Q'}$ and
$\u(T')$ is in the space generated by the characteristic
functions of the states of $\A$ since
$\u(\{1,2,6\})=\u(\{1,2,5\})-\u(\{1,4,5\})+\u(\{1,4,6\})$.
Thus $S'$ is completely reducible althought it is not birecurrent.
To see that $S'$ is not birecurrent, one may either use the fact that
$T'$ is not saturated by any word of minimal rank 
and apply Theorem~\ref{theoremRev} or
compute directly $\tilde{\A'}^\delta$ (see Figure \ref{figureA'}).
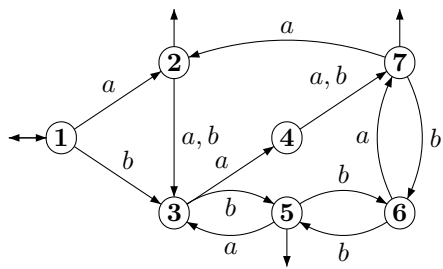
\begin{figure}[hbt]
\centering
\gasset{Nadjust=wh}
\begin{picture}(60,30)(0,-5)
\node[Nmarks=if,fangle=180](1)(0,10){$\mathbf{1}$}
\node[Nmarks=f,fangle=90](2)(15,20){$\mathbf{2}$}
\node(3)(15,0){$\mathbf{3}$}
\node(4)(30,10){$\mathbf{4}$}\node[Nmarks=f,fangle=-90](5)(30,0){$\mathbf{5}$}
\node(6)(45,0){$\mathbf{6}$}\node[Nmarks=f,fangle=90](7)(45,20){$\mathbf{7}$}

\drawedge(1,2){$a$}\drawedge(1,3){$b$}
\drawedge(2,3){$a,b$}
\drawedge(3,4){$a$}\drawedge[curvedepth=3,ELside=r](3,5){$b$}\drawedge[curvedepth=3](5,3){$a$}
\drawedge(4,7){$a,b$}\drawedge[curvedepth=3](5,6){$b$}\drawedge[curvedepth=3](6,5){$b$}
\drawedge[curvedepth=3](6,7){$a$}\drawedge[curvedepth=3](7,6){$b$}
\drawedge[curvedepth=-3,ELside=r](7,2){$a$}
\end{picture}
\caption{The automaton $\tilde{\A'}$}\label{figureA'}
\end{figure}
Note that the hypothesis of Theorem~\ref{theoremCharactCompleteRed}
is satisfied because it is already satisfied by the minimal
representation of $\u(S)$, which is the same as that
of $\u(S')$  except for the vectors 
$\gamma=\u(T)$ and $\gamma'=\u(T')$.
\end{example}
We finally deduce from Corollary~\ref{corollary2}
 the following result, originally  proved in~\cite{Reutenauer1981} (see also
\cite[Theorem 14.7.7]{BerstelPerrinReutenauer2009}).
\begin{corollary}
Let $X$ be a recognizable maximal code. If $X^*$ is completely reducible,
then $X$ is bifix.
\end{corollary}
\begin{proof}
Let $\A$ be the minimal automaton of $X^*$. Since $X$
is a recognizable code, it is thin by~\cite[Proposition 2.5.20]{BerstelPerrinReutenauer2009}. Since $X$ 
is a thin maximal code, $X^*$ is dense by
\cite[Theorem 2.5.5]{BerstelPerrinReutenauer2009}.
As for any recognizable code, there is a trim
unambiguous finite automaton $\A=(Q,i,i)$
recognizing $X^*$ (see~\cite[Proposition 4.1.2]{BerstelPerrinReutenauer2009}).
Since $\A$ is trim with a unique initial state equal to the unique terminal
state, it is strongly connected.
By Corollary~\ref{corollary2}, $\varphi_\A(X^*)$ meets every $\GH$-class
of the minimal ideal of $\varphi_\A(A^*)$. This implies that
$X^*$ is recurrent and thus $X$ is prefix by~\cite[Proposition 3.3.11]{BerstelPerrinReutenauer2009}.
Symmetrically, $X$ is suffix, whence the conclusion.
\end{proof}

\subsection{A characterization of completely reducible sets}
\label{sectionCharacterizationCompletelyReducible}
In this section, we consider rational series over $\Q$
whose set of coefficients is finite. It is known that
the minimal representation $(\lambda,\mu,\gamma)$ of such a series 
satisfies the following finiteness property: the matrix monoid
$\mu(A^*)$ is finite. This follows from a theorem of Sch\"utzenberger,
see \cite[Corollary 2.3]{BerstelReutenauer2011}.

Moreover, such a series is a linear combination over $\Q$ of characteristic series
of rational languages. This follows from the fact that for any
$\alpha\in\Q$, the set $\{w\in A^*\mid (S,w)=\alpha\}$ is rational.
For this, see \cite[Theorem 2.10]{BerstelReutenauer2011},
another theorem of Sch\"utzenberger.

For a rational series $S$ whose set of coefficients is finite, we
may construct a special kind of deterministic automaton,
called an automaton \emph{with scalar output function}.
It is a deterministic automaton $\A=(Q,i,\tau)$ where
$\tau:Q\rightarrow \Q$ is a mapping called the \emph{terminal function},
which recognizes $S$ in the following sense. For each word $w$,
one has $(S,w)=\tau(i\cdot w)$. In other words, one reads $w$ on the automaton, 
starting from the initial state $i$ and one reaches a state $q=i\cdot w$.
The coefficient $(S,w)$ of $w$ in $S$ is $\tau(q)$. If there
is no path from $i$ labeled $w$, we let $(S,w)=0$.

The notions
of trim automaton and of minimal automaton extend easily to these automata.

 For a word $u\in A^*$
and a series $S$,
we denote
here $u^{-1}S$ the series defined by $(u^{-1}S,v)=(S,uv)$
(this series is denoted $S\cdot u$ in Section~\ref{sectionFormalSeries})
and symmetrically $Su^{-1}$ the series defined
by $(Su^{-1},v)=(S,vu)$.

There is also a Nerode criterium for these series. Indeed, a series $S$
is rational (and has a finite set of coefficients)
if and only if the set $\{u^{-1}S\mid u\in A^*\}$ is finite.

We say that a rational series with a finite number of coefficients
is \emph{recurrent} if its minimal automaton with scalar output function
is strongly connected. 

Clearly, a series $S$ is recurrent if and
only if it is recognized by a strongly connected automaton
with scalar output function. Moreover, if
$S$ is recurrent, then so are $u^{-1}S$ and $Su^{-1}$ for any
word $u$, since they are recognized by the minimal
automaton with output function of $S$.

It is called \emph{birecurrent} if $S$ and $\tilde{S}$
are both recurrent, where $\tilde{S}$ is the series such that
$(\tilde{S},w)=(S,\tilde{w})$ for all $w\in A^*$.
\begin{proposition}\label{propositionCR1}
Let $S$ be a birecurrent series. Then $S$ is a linear combination
over $\Q$ of characteristic series of birecurrent sets.
\end{proposition}
\begin{proof}
Let $I=\{\alpha\in \Q\setminus 0\mid (S,w)=\alpha\mbox{ for some $w\in A^*$}\}$.
Then $I$ is finite. It is enough to show that for any $\alpha\in I$, the set
$L=\{w\in A^*\mid(S,w)=\alpha\}$ is recurrent.

Let $(Q,i,\tau)$ be a strongly connected automaton with scalar output function
recognizing $S$. Let
$(Q,i,T)$ be the automaton defined by $T=\{q\in Q\mid\tau(q)=\alpha\}$.
Then this automaton recognizes $L$. It is deterministic and strongly
connected and thus $L$ is recurrent.
\end{proof}
\begin{proposition}\label{propositionCR2}
Let $S$ be a completely reducible series. Then $S$ is a linear combination
over $\Q$ of birecurrent series.
\end{proposition}
\begin{proof}
Let $(\lambda,\mu,\gamma)$ be a minimal representation of $S$. Since
the set of coefficients of $S$ is finite, we know that $\mu(A^*)$ is a finite
monoid.

Since the representation is completely reducible, it is isomorphic
to a direct sum  of  representations $(\lambda_i,\mu_i,\gamma_i)$
which are irreducible over $\Q$.

It follows that $S$ is the corresponding sum of irreducible series. We may
therefore assume that $(\lambda,\mu,\gamma)$ is irreducible. Note
that $\mu(A^*)$ is finite.

The set $R(\lambda)=\{\lambda\mu(u)\mid u\in A^*\}$ is finite. One obtains a right action
of $A^*$ on this set by $\lambda\mu(u)\cdot w=\lambda\mu(uw)$. There
exists  therefore a word $u$ such that the set
$R(\lambda\mu(u))$ is a minimal invariant subset for this action.

Similarly, $A^*$ acts on the left
on the finite set $L(\gamma)=\{\mu(v)\gamma\mid v\in A^*\}$ by 
$w\cdot\mu(v)\gamma=\mu(wv)\gamma$. There exists similarly a word $v$
such that the set $L(\mu(v)\gamma)$ is minimal for this action.

Consider the linear representation $(\lambda\mu(u),\mu,\mu(v)\gamma)$.
It recognizes a rational series, call it $T$, whose set
of coefficients is finite since $\mu(A^*)$ is finite.

For this series $T$, we may construct the following deterministic automaton
with scalar output function.
Its set of states is $R(\lambda\mu(u))$, its initial
state is $\lambda\mu(u)$ and its terminal function is
defined by $\tau(\lambda')=\lambda'\mu(v)\gamma$ for any state
$\lambda'\in R(\lambda\mu(u))$.

Clearly this automaton recognizes $T$. Since its set of states
is a minimal invariant subset for the right action of $A^*$,
this automaton
is strongly connected. Thus $T$ is recurrent.

By symmetry, $\tilde{T}$ is also recurrent. Hence $T$ is birecurrent.

Turning back to the representation
$(\lambda,\mu,\gamma)$, there exists, since it is irreducible,
polynomials $X,Y\in \Q\langle A\rangle$ 
such that $\lambda\mu(u)\mu(X)=\lambda$ and $\mu(Y)\mu(v)\gamma=\gamma$
where the monoid morphism $\mu:A^*\rightarrow \Q^{n\times n}$ is extended to an
algebra morphism still denoted $\mu$ from $\Q\langle A\rangle$
into $\Q^{n\times n}$.

Finally, for any word $w$, one has
\begin{eqnarray*}
(S,w)=&=&\lambda\mu(w)\gamma\\
&=&\lambda\mu(u)\mu(X)\mu(w)\mu(Y)\mu(v)\gamma\\
&=&\lambda\mu(u)\left(\sum_{x\in A^*}(X,x)\mu(x)\right)\mu(w)\left(\sum_{y\in A^*}(Y,y)\mu(y)\right)\mu(v)\gamma\\
&=&\sum_{x,y}(X,x)(Y,y)\lambda\mu(u)\mu(x)\mu(w)\mu(y)\mu(v)\gamma\\
\end{eqnarray*}
Thus, since $T=\sum_w \lambda \mu(u)\mu(w)\mu(v) \gamma$, 
we have $x^{-1}Ty^{-1} = \sum_w \lambda \mu(u)\mu(x)\mu(w)\mu(y)\mu(v) \gamma$, 
and therefore $S=\sum_{x,y}(X,x)(Y,y)x^{-1}Ty^{-1}$. Since $T$ is birecurrent,
each $x^{-1}Ty^{-1}$ is birecurrent. Hence $S$ is a linear
combination of birecurrent series.
\end{proof}
\begin{proposition}\label{propositionCR3}
Let $S\in \Q\<<A>>$ be a rational series whose set of coefficients is finite. Then $S$ is completely reducible if and only if it is a linear
combination over $\Q$ of characteristic series
of birecurrent sets.
\end{proposition}
\begin{proof}
Suppose that $S$ is completely reducible. Then the conclusion follows by Propositions \ref{propositionCR1} and \ref{propositionCR2}.

Conversely, suppose that $S$ is a linear combination over $\Q$
of characteristic series of birecurrent sets. 
By Corollary \ref{corollaryCompleteRed}, each such series is completely
reducible. By~\cite[Proposition 4.1]{Perrin2013}, a linear combination over $\Q$
of completely reducible series is completely reducible. Thus $S$
is completely reducible.
\end{proof}
We thus obtain as a main result of this section.
\begin{theorem}\label{theoremCR3}
A language is completely reducible if and only if its
characteristic series is $\Q$-linear combination of
characteristic series of birecurrent languages.
\end{theorem}
The following example shows that the linear combination
need not have coefficients in $\mathbb Z$.
\begin{example}
Let $\A=(Q,i,T)$ be the automaton represented in Figure~\ref{figureAutomatonQlin}
on the left with its reversal on the right.
\begin{figure}[hbt]
\centering
\gasset{Nadjust=wh}
\begin{picture}(70,30)(0,-6)
\put(0,0){
\begin{picture}(20,30)
\node[Nmarks=if,fangle=180](1)(0,15){$1$}\node(2)(20,15){$2$}\node(3)(10,0){$3$}

\drawloop(1){$a$}\drawedge(1,3){$c$}\drawloop[loopangle=-90,ELside=r](3){$c$}
\drawedge(3,2){$b$}\drawloop(2){$b$}\drawedge(2,1){$a$}
\end{picture}
}
\put(35,0){
\begin{picture}(20,40)
\node[Nmarks=if,fangle=180](1)(5,15){$1$}
\node[Nmarks=f,fangle=90](12)(20,15){$1,2$}\node(23)(40,15){$2,3$}\node[Nmarks=f,fangle=90](13)(30,0){$1,3$}

\drawedge(1,12){$a$}
\drawloop(12){$a$}\drawedge(12,23){$b$}\drawloop(23){$b$}
\drawedge(23,13){$c$}\drawloop[loopangle=-90,ELside=r](13){$c$}\drawedge(13,12){$a$}
\end{picture}
}
\end{picture}
\caption{A finite automaton and its deterministic reversal.}\label{figureAutomatonQlin}
\end{figure}
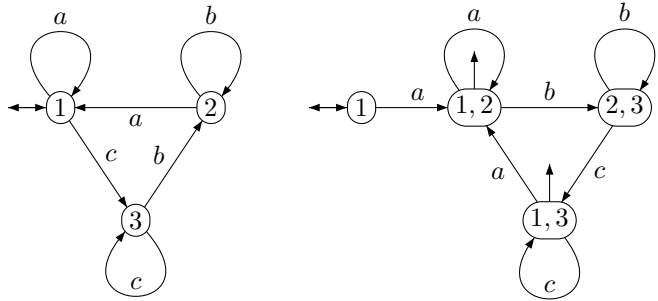

Let $X$ be the set recognized by $\A$. Since $\A$
is strongly connected, $X$ is recurrent. 
Since $\tilde{A}^\delta$ is not strongly connected,
$X$ is not birecurrent. Let $X_i$ be the
set recognized using, instead of $T=\{1\}$, the set of terminal states $T_i$
for $1\le i\le 3$ with $T_1=\{2,3\}$, $T_2=\{1,3\}$ and $T_3=\{1,2\}$.
Since these sets are saturated respectively by $b$, $c$ and $a$,
which have rank $1$,
the sets $X_i$ are birecurrent (this can also
be seen easily in Figure~\ref{figureAutomatonQlin}).
Since $\u(T)=\u(T_3)-\u(T_1)-\u(T_2)$, we have
$X=\frac{1}{2}(X_3-X_1+X_2)$. Note that, by Proposition~\ref{propositionNumberIrred}, the linear representation $(\lambda,\mu,\gamma)$ associated to the automaton
$\A$ is irreducible. Indeed, this representation is clearly minimal
and since the letters are of rank $1$, the Suschkevitch
group of the monoid $\mu(A^*)$ is trivial. Thus we obtain a decomposition
of an irreducible set $X$ as a $\Q$-linear combination of
birecurrent sets. We conjecture that $X$ cannot be obtained
as a linear combination of birecurrent sets with coefficients
in $\Z$.
\end{example}

\section{Unambiguous automata}\label{sectionUnambiguous}
In this section, we generalize some of the results concerning
birecurrent sets using unambiguous automata, which
are nondeterministic automata
closely linked with linear representation
of series. For an introduction
to this class of automata, see~\cite{BerstelPerrinReutenauer2009}.
\subsection{Unambiguous automata and linear representations}
An automaton $\A=(Q,I,T)$ is \emph{unambiguous} if for any word $w\in A^*$
there is at most one path from a state $i\in I$ to a state $t\in T$
labeled $w$.

A trim unambiguous automaton has the following property (used in~\cite{BerstelPerrinReutenauer2009}
as a definition of unambiguous automata).
For every word
$w\in A^*$ and every pair $(p,q)\in Q$ of states there is at most one
path from $p$ to $q$ labeled $w$. Indeed, since $\A$ is trim,
there is a path $i\edge{u}p$ with $i\in I$ and a path $q\edge{v}t$
with $t\in T$. Since there is at most one path from $i$ to $t$
labeled $uwv$, there is at most one path from $p$ to $q$ labeled $w$.

Clearly a deterministic automaton is unambiguous.

\begin{example}\label{exampleUnambig1}
The automaton $\A=(Q,1,1)$ represented in Figure~\ref{figureNonAmbig} is unambiguous.
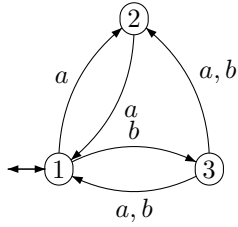
\begin{figure}[hbt]
\centering\gasset{Nadjust=wh}
\begin{picture}(20,25)(0,-5)
\node[Nmarks=if,fangle=180](1)(0,0){$1$}
\node(2)(10,20){$2$}
\node(3)(20,0){$3$}

\drawedge[curvedepth=3](1,2){$a$}\drawedge[curvedepth=3](2,1){$a$}
\drawedge[curvedepth=3](1,3){$b$}
\drawedge[curvedepth=3](3,1){$a,b$}\drawedge[curvedepth=-3,ELside=r](3,2){$a,b$}
\end{picture}
\caption{An unambiguous automaton.}\label{figureNonAmbig}
\end{figure}
One can check this by computing the automaton of pairs and by checking
that there is no path from a pair $(p,p)$ to a pair $(q,q)$ using a pair
$(r,s)$ with $r\ne s$.
\end{example}

To every unambiguous automaton $\A=(Q,I,T)$ recognizing a set $S$, we may associate a linear
representation recognizing the characteristic series $\u(S)$
of the set $S$. We consider,
as in Section~\ref{sectionPreliminaries}
a field $K$ 
and the vector space $V=K^Q$.
Consider the representation $(\lambda,\mu,\gamma)$
where $\lambda\in V$ is the characteristic function $\u(I)$ of $I$
considered as a row vector,
$\gamma=\u(T)$ considered as a column vector and for $a\in A$,
$\mu(a)$ is the $Q\times Q$-matrix $\mu_\A(a)$ defined by 
\begin{displaymath}
(p,\mu_\A(a),q)=\begin{cases}1&\mbox{if $p\edge{a}q$}\\0&\mbox{otherwise.}
\end{cases}
\end{displaymath}
Then $\mu$ extends to a morphism from $A^*$ into $\End(V)$ and $\lambda\mu(w)\gamma=(\u(S),w)$ for every $w\in A^*$. When $\A$ is trim, all the matrices $\mu_\A(w)$ for $w\in A^*$
have coefficients $0$ or $1$.
\begin{example}
The linear representation corresponding to the unambiguous
automaton of Figure~\ref{figureNonAmbig} is
\begin{displaymath}
\lambda=[1\ 0\ 0],\quad \mu(a)=\begin{bmatrix}0&1&0\\1&0&0\\1&1&0\end{bmatrix},
\quad \mu(b)=\begin{bmatrix}0&0&1\\0&0&0\\1&1&0\end{bmatrix},\quad
\gamma=\begin{bmatrix}1\\ 0\\ 0\end{bmatrix}.
\end{displaymath}
\end{example}
For a set $U\subset Q$ of states and a word $w$, we denote
\begin{eqnarray*}
U\cdot w&=&\{q\in Q\mid u\edge{w}q\mbox{ for some }u\in U\},\\
w\cdot U&=&\{q\in Q\mid q\edge{w}u\mbox{ for some }u\in U\}.
\end{eqnarray*}

Let $\A$ be an unambiguous automaton.
Recall that we denote by $\A^\delta$ the 
 determinization of $\A$. Thus the states of $\A^\delta$
are the nonempty sets $I\cdot w$ for $w\in A^*$.
The initial state is $I$ and the set of terminal states
is the set of $U=I\cdot w$ such that $U\cap T\ne\emptyset$.
Note that since $\A$ is unambiguous, $U\cap T$ contains at
most one element and that $\A$ and $\A^\delta$ recognize the same set of words.

When $\A$ is trim, the action of $A^*$ on the states of $\A^\delta$ is the
same as the right multiplication by the matrices $\mu_\A(w)$
on the characteristic vectors. Indeed, one has $U=I\cdot w$
if and only if $\u(U)=\u(I)\mu_\A(w)$.

Also recall that we denote by $\tilde{\A}^\delta$ the deterministic 
reversal of the automaton $\A$.

\begin{example}\label{exampleUnambig2}
Let $\A$ be the unambiguous automaton of Figure~\ref{figureNonAmbig}.
The automaton $\A^\delta$ is represented in Figure~\ref{figureLambdaA}
on the left and the automaton $\tilde{\A}^\delta$ on the right.
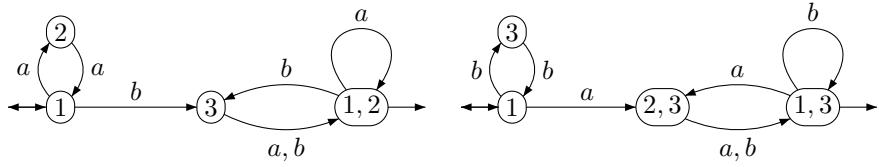
\begin{figure}[hbt]
\centering\gasset{Nadjust=wh}
\begin{picture}(100,20)
\put(0,0){
\begin{picture}(40,15)(0,-5)
\node[Nmarks=if,fangle=180](1)(0,0){$1$}
\node(2)(0,10){$2$}
\node(3)(20,0){$3$}
\node[Nmarks=f](4)(40,0){$1,2$}

\drawedge[curvedepth=3](1,2){$a$}\drawedge[curvedepth=3](2,1){$a$}
\drawedge(1,3){$b$}
\drawedge[curvedepth=-3,ELside=r](3,4){$a,b$}
\drawedge[curvedepth=-3,ELside=r](4,3){$b$}
\drawloop(4){$a$}
\end{picture}
}
\put(60,0){
\begin{picture}(40,15)(0,-5)
\node[Nmarks=if,fangle=180](1)(0,0){$1$}
\node(3)(0,10){$3$}
\node(23)(20,0){$2,3$}
\node[Nmarks=f](13)(40,0){$1,3$}
\drawedge[curvedepth=3](1,3){$b$}\drawedge[curvedepth=3](3,1){$b$}
\drawedge(1,23){$a$}
\drawedge[curvedepth=-3,ELside=r](23,13){$a,b$}
\drawedge[curvedepth=-3,ELside=r](13,23){$a$}
\drawloop(13){$b$}
\end{picture}
}
\end{picture}
\caption{The automata $\A^\delta$ and $\tilde{\A}^\delta$.}\label{figureLambdaA}
\end{figure}
\end{example}

The \emph{rank} of a word $w$ with respect to an unambiguous
automaton $\A$ is the rank of the linear map $\mu_\A(w)$.
This definition is consistent with the one given
for a deterministic automaton in Section~\ref{sectionBirecurrent}.
Indeed, when $\A$ is deterministic, the matrix $\mu_\A(w)$ is the
matrix of a partial map from $Q$ into itself and the
rank of $\mu_\A(w)$ is equal to the rank of the map.
The definition of rank for an unambiguous automaton
given in~\cite{BerstelPerrinReutenauer2009} is different
but equivalent to this one (see \cite[Exercise 9.3.2]{BerstelPerrinReutenauer2009}).

Let $\A=(Q,I,T)$ be a trim unambiguous automaton.
Then the monoid $\mu_\A(A^*)$ is formed of $\{0,1\}$-matrices and thus
it is finite.
As for a deterministic automaton, the set of elements of minimal
nonzero rank of the monoid $M=\mu_\A(A^*)$ is the unique $0$-minimal
ideal $J$ of $M$. It is the union of all $0$-minimal
right (resp. left) ideals. It is formed of a regular $\GD$-class, plus possibly
$0$. Each $\GR$-class of $J\setminus \{0\}$ is formed of elements which have
the same set of rows. Each $\GH$-class of $J\setminus \{0\}$ which is a group is a transitive
permutation group on the common set of rows of its elements.

\begin{example}\label{exampleUnambig3}
Let $\A$ be the automaton represented in Figure~\ref{figureNonAmbig}.
The minimal ideal $J$ of the monoid $M=\mu_\A(A^*)$ is represented in Figure~\ref{figureMinimalIdealUnambiguous}.
\begin{figure}[hbt]

\begin{displaymath}
\def\rb{\hspace{2pt}\raisebox{0.8ex}{*}}\def\vh{\vphantom{\biggl(}}
    \begin{array}%
    {r|@{}l@{}c|@{}l@{}c|}%
    \multicolumn{1}{r}{}&\multicolumn{2}{c}{3}&\multicolumn{2}{c}{1,2}\\
    \cline{2-5}
    2,3& \vh\rb &ab &\vh\rb  &aba \\
    \cline{2-5}
    1,3&\vh\rb &bab &\vh\rb& ba \\
    \cline{2-5}
    \end{array}
\end{displaymath}
\caption{The minimal ideal of $M$.}\label{figureMinimalIdealUnambiguous}
\end{figure}
There is no word of rank $0$ and the minimal rank is $1$. For example, we have
\begin{displaymath}
\mu_\A(a)=\begin{bmatrix}0&1&0\\1&0&0\\1&1&0\end{bmatrix},\quad
\mu_\A(ab)=\begin{bmatrix}0&0&0\\0&0&1\\0&0&1\end{bmatrix},
\end{displaymath}
the first one being of rank $2$ and the second of rank $1$.
We represent for each $\GR$-class of $J$ the set of rows of the elements
(representing a row as a set of states)
and for each $\GL$-class its set of columns.
\end{example}

\subsection{Unambiguous automata and birecurrent sets of finite type}

The following result is a generalization of Theorem~\ref{theoremRev}
which gives a sufficient condition for the set
recognized by an unambiguous automaton to be recurrent.
The proof is quite similar.
\begin{theorem}\label{theoremCR}
Let $\A=(Q,I,T)$
 be a strongly connected unambiguous finite automaton.
The automaton $\A^\delta$ is strongly connected
 if and only if
there is a word $x$ of minimal nonzero rank such that
$I\cdot w=I$.
\end{theorem}
\begin{proof}
Assume first that $\A^\delta$ is strongly connected. Let
$x$ be a word of minimal nonzero rank. Since $x$ has nonzero rank
there exist $p,q$ such that $p\edge{x}q$. Since $\A$ is strongly
connected, there is a word $u$ such that $i\edge{u}p$ for some $i\in I$.
Then $y=ux$ is a word of minimal nonzero rank such that $I\cdot y\ne\emptyset$.
Thus $I\cdot y$ is a state of $\A^\delta$.
Since $\A^\delta$ is strongly connected, there is a word $z$ such that
$I\cdot yz=I$. Since $yz$ has minimal nonzero rank, the conclusion follows.

Conversely, assume that $I=I\cdot w$ with $w$ of minimal nonzero rank. 
Consider a state $U$ of $\A^\delta$. Then $U=I\cdot u$
for some
$u\in A^*$  such that $I\cdot u\ne\emptyset$. Then $I\cdot wu=I\cdot u$
and thus $wu$ is a word of minimal nonzero rank. Set $\mu=\mu_\A$.
Since the right ideal
generated by $\mu(w)$ is $0$-minimal, there is a word $v\in A^*$ such that
$\mu(wuv)=\mu(w)$. Let $e=\mu(uv)^n$ with $n\ge 1$ be the idempotent
which is a power of $\mu(uv)$. Since $\mu(wuv)=\mu(w)$,
we have $\mu(w)e=\mu(w)$. Then $\u(I)e=\u(I\cdot w)e=\u(I)\mu(w)e=\u(I)\mu(w)=\u(I)$.
Set $x=v(uv)^{n-1}$. Then $\u(I)\mu(ux)=\u(I)e=\u(I)$ and thus $I\cdot ux=I$.
This shows that $I$ and $I\cdot u$ belong to the same strongly
connected component and thus that $\A^\delta$ is strongly connected.
\end{proof}
A symmetric result holds for the automaton $\tilde{\A}^\delta$, which is strongly
conected if and only if there is a word $x$ of minimal nonzero rank
such that $x\cdot T=T$.

Theorem~\ref{theoremRev} follows easily from the symmetric version
of Theorem~\ref{theoremCR}. In fact, assume that
$\A$ is deterministic. A set $U$ is saturated by a word of minimal nonzero rank
if and only if there is a word $x$ of minimal nonzero rank such that $x\cdot U=U$.
Indeed, the condition is sufficient. Conversely, if $U=x\cdot V$ for some
$V\subset Q$ and some $x$ of minimal nonzero rank, let $y$ be a word
such that $\varphi_\A(y)$ is a nonzero idempotent in the right
ideal generated by $\varphi_\A(x)$. Then $\varphi_\A(yx)=\varphi_\A(x)$
and thus $y\cdot U=yx\cdot V=x\cdot V=U$. 

\begin{corollary}\label{corollaryCR}
Let $\A=(Q,I,T)$ be a strongly connected unambiguous finite automaton recognizing a set $S$.
If there are words $x,y$ of $0$-minimal rank such that $I\cdot x=I$ and $y\cdot T=T$, then $S$
is birecurrent.
\end{corollary}
\begin{proof}
By Theorem \ref{theoremCR}, since $I\cdot x=I$, the automaton $\A^\delta$ is strongly connected and conversely.
Symmetrically, since $y\cdot T=T$, the automaton
$\tilde{\A}^\delta$ is strongly connected and conversely. 
\end{proof}

 Iterating the construction of Section~\ref{sectionConstruction},
one obtains examples of birecurrent sets defined by unambiguous
automata as in
Corollary~\ref{corollaryCR}.
The iteration relies on the
 following result  from~\cite{Vincent1985} (see~\cite[Exercise 14.1.9]{BerstelPerrinReutenauer2009}) where we use again the notation
$\delta_w$ and $\gamma_w$ introduced after Proposition~\ref{propDelta}.
\begin{proposition}\label{propositionVincent}
Let $Z$ be a finite maximal prefix code and let $w^2$ be a pure
square for $Z$. Then $w^4$ is a pure square for $X=\delta_w(Z)$
and for $Y=\gamma_w(Z)$.
The sets $G'=X(w^2)^{-1}$ and $D'=(w^2)^{-1}X$ satisfy
\begin{displaymath}
\u(G')-1=(1+w)(\u(G)-1),\quad \u(D')-1=(1+w)(\u(D)-1)
\end{displaymath}
and the sets $G''=Y(w^2)^{-1}$ and $D''=(w^2)^{-1}Y$ satisfy
\begin{displaymath}
\u(G'')-1=(\u(G)-1)(1+w),\quad \u(D'')-1=(\u(D)-1)(1+w).
\end{displaymath}
\end{proposition}
We deduce from Proposition~\ref{propositionVincent} the following result.
\begin{theorem}\label{theoremVincent}
Let $Z$ be a finite maximal bifix code, let $w$ be a pure square
for $Z$ and let $U=\gamma_{w^2}(\delta_w(Z))$. Then
$\{\varepsilon,w^2\}U^*\{\varepsilon,w\}$ is a birecurrent
set of finite type.
\end{theorem}
\begin{proof}
Set $X=\delta_w(Z)$, $Y=\gamma_w(Z)$, $L=\delta_{w^2}(X)$
and $R=\gamma_{w^2}(Y)$. Since $Z$ is a finite maximal prefix code,
$X$ and $L$ are finite maximal prefix codes. Symmetrically,
since $Z$ is a finite maximal suffix code, $Y$ and $R$ are
finite maximal suffix codes.

Let $G,D,G',D'$ be as in Proposition~\ref{propositionVincent}.
We have by Equation~\eqref{eqDelta2}
\begin{displaymath}
\u(U)-1=(\u(X)-1+(\u(G')-1)w^2(\u(D')-1))(1+w^2)
\end{displaymath}
and by Equation~\eqref{eqDelta}
\begin{displaymath}
\u(L)-1=(1+w^2)(\u(X)-1+(\u(G')-1)w^2(\u(D')-1)).
\end{displaymath}
Thus
\begin{eqnarray*}
(1+w^2)(\u(U)-1)=(\u(L)-1)(1+w^2).
\end{eqnarray*}
Therefore $(1+w^2)\u(U)^*(1+w)=\u(L)^*(1+w^2)(1+w)$.
This shows that $V=\{\varepsilon,w^2\}U^*\{\varepsilon,w\}$ is recurrent with a finite left root.
Similarly
\begin{displaymath}
\u(R)-1=(\u(Y)-1+(\u(G'')-1)w^2(\u(D'')-1))(1+w^2).
\end{displaymath}
Since $(\u(X)-1)(1+w)=(1+w)(\u(Y)-1)$ and $(\u(G')-1)w^2(\u(D')-1)(1+w)=(1+w)(\u(G'')-1)w^2(\u(D'')-1)$,
we have
\begin{displaymath}
(1+w)(\u(R)-1)=(\u(X)-1+(\u(G')-1)w^2(\u(D')-1))(1+w)(1+w^2)
\end{displaymath}
showing that 
\begin{displaymath}
(1+w^2)(1+w)(\u(R)-1)=(\u(L)-1)(1+w)(1+w^2).
\end{displaymath}
This implies that $V$ is birecurrent and that its right root is finite.
\end{proof}
\begin{example}
Let $\A=(Q,I,T)$ be the nondeterministic automaton with transitions given in Table~\ref{tableTransVincent} with $I=\{1,14\}$ and $T=\{1,4\}$.
\begin{table}[hbt]
\begin{displaymath}
\begin{array}{c|c|c|c|c|c|c|c|c|c|c|c|c|c|c|c|c|c|}\cline{2-18}
   &1   &2   &3   &4  &5   &6 &7&8&9   &10  &11  &12&13&14&15&16&17\\\cline{2-18}
a  &2   &6,10&4   &2  &6   &  &8& &10  &1,14&12  &  &14&  &16&  &1\\\cline{2-18}
b  &3   &1,14&7,11&5  &7,11&7 & &9&    &11  &1,14&13&  &15&  &17&\\\cline{2-18}
a^2&6,10&1,14&2   &6,10&   &  & & &1,14&2   &    &  &  &  &  &  &2\\\cline{2-18}
\end{array}
\end{displaymath}
\caption{The transitions of the automaton $\A$.}\label{tableTransVincent}
\end{table}
One may verify that this automaton is unambiguous. The word $a^2$ has rank $3$
since the matrix $\mu_\A(a^2)$ has $3$ distinct nonzero rows which are
the characteristic vectors of $\{1,14\}$, $\{6,10\}$ and $\{2\}$. 

This
rank is minimal as one can check by computing the $9$ images of this $3$-element
set by the action of the letters. Since $\u(I)$ is a row of $\mu(a^2)$, the hypothesis
of Theorem~\ref{theoremCR} are satisfied. The automaton $\A^\delta$ is strongly
connected and has also $17$ states represented in Table~\ref{tableStatesLambda}.
\def\mbf{\mathbf}
\begin{table}[hbt]
\begin{displaymath}
\begin{array}{cc}
\begin{array}{|c|c|c|c|c|c|c|c|c|c|c|c|}
\mbf{1}&\mbf{2}&\mbf{3}&\mbf{4}&\mbf{5}&\mbf{6}&\mbf{7}&\mbf{8}&\mbf{9}&\mbf{10}&\mbf{11}&\mbf{12}\\\hline
1,14   &2      &3,15   &6,10   &4,16   &7,11   &5,17   &8,12   &1,6    &9,13    &3,7     &10,14   
\end{array}
\\
\qquad \quad\quad
\begin{array}{|c|c|c|c|c|c|}
\mbf{13}&\mbf{14}&\mbf{15}&\mbf{16}&\mbf{17}\\\hline
4,8     &11,15   &5,9     &12,16   &13,17
\end{array}
\end{array}
\end{displaymath}
\caption{The states of the automaton $\A^\delta$.}\label{tableStatesLambda}
\end{table}
The transitions of $\A^\delta$ are represented in Table~\ref{tableTransLambda}.
\begin{table}[hbt]
\begin{displaymath}
\begin{array}{c|c|c|c|c|c|c|c|c|c|c|c|c|c|c|c|c|c|}\cline{2-18}
   &\mbf{1}&\mbf{2}&\mbf{3}&\mbf{4}&\mbf{5}&\mbf{6}&\mbf{7}&\mbf{8} &\mbf{9} &\mbf{10}&\mbf{11}&\mbf{12}&\mbf{13}&\mbf{14}&\mbf{15}&\mbf{16}&\mbf{17}\\\cline{2-18}
a  &\mbf{2}&\mbf{4}&\mbf{5}&\mbf{1}&\mbf{2}&\mbf{8}&\mbf{9}&\mbf{1} &\mbf{2} &\mbf{12}&\mbf{13}&\mbf{1} &\mbf{2} &\mbf{1} &\mbf{4} &\mbf{1} &\mbf{1}\\\cline{2-18}
b  &\mbf{3}&\mbf{1}&\mbf{6}&\mbf{6}&\mbf{7}&\mbf{1}&\mbf{6}&\mbf{10}&\mbf{11}&\mbf{1} &\mbf{6} &\mbf{14}&\mbf{15}&\mbf{6} &\mbf{6} &\mbf{17}&\mbf{1}\\\cline{2-18}
\end{array}
\end{displaymath}
\caption{The transitions of the automaton $\A^\delta$.}\label{tableTransLambda}
\end{table}
The set $T=\{1,4\}$ is a column of $\mu_\A(a^2)$ (the columns of index $6$ and $10$).
Thus, by the symmetric of Theorem \ref{theoremCR}, the automaton $\tilde{\A}^\delta$ is also
strongly connected. Thus the set $S$ recognized by $\A$ is birecurrent.

The set $S$ is an instance of Theorem~\ref{theoremVincent}.
To see this, we start with the finite maximal bifix code $\tilde{Z}$ represented in Figure~\ref{fig3_03}
on the right. We choose the word $x=(ba)^2$ which is a pure square for $\tilde{Z}$.
Then $\tilde{Y}=\delta_w(\tilde{Z})$ is the maximal prefix code represented in Figure~\ref{figureRightRoot}.
The word $w^2=(ba)^4$ is a pure square for $\tilde{Y}$ and the set $X=\gamma_{w^2}$ is
a maximal code generating the set recognized by the automaton $\A$ with $1$
as initial and terminal state.
\end{example}

\bibliography{birecurrent}
\bibliographystyle{plain}
\end{document}